\documentclass[11pt, oneside]{article}
\usepackage[margin=2cm]{geometry}                		
\usepackage[parfill]{parskip}    		
\usepackage{graphicx}			
\usepackage[utf8]{inputenc}
\usepackage[english]{babel} 
\usepackage{amsmath,amssymb,amsthm,bm} 
\usepackage[margin=2cm]{geometry}
\usepackage[color=yellow]{todonotes}
\usepackage{mathrsfs}
\usepackage{url}	
\usepackage{esint}
\usepackage[colorlinks]{hyperref}
\usepackage{enumerate}
\usepackage{outlines}[enumerate]
\usepackage{framed,comment,enumerate}
\usepackage{upgreek}
\usepackage{pgfplots}
\usepackage{float}
\usepackage{tikz,tikz-cd}
\usepackage{rotating}
\usepackage{graphicx}
\usepackage{caption}
\usepackage{multicol}
\usepackage{cancel}
\usepackage{subcaption}
\usepackage[title]{appendix}
\usepackage{tikz-cd}
\usepackage{pgf, pgffor}
\usepackage{subcaption}
\extrafloats{100}
\numberwithin{equation}{section}

\newtheorem{theorem}{Theorem}[section]

\newtheorem{proposition}{Proposition}[section]

\theoremstyle{definition}

\DeclareFontFamily{U}{MnSymbolC}{}
\DeclareSymbolFont{MnSyC}{U}{MnSymbolC}{m}{n}
\DeclareFontShape{U}{MnSymbolC}{m}{n}{
    <-6>  MnSymbolC5
   <6-7>  MnSymbolC6
   <7-8>  MnSymbolC7
   <8-9>  MnSymbolC8
   <9-10> MnSymbolC9
  <10-12> MnSymbolC10
  <12->   MnSymbolC12}{}
\DeclareMathSymbol{\intprod}{\mathbin}{MnSyC}{'270}

\newcommand{\mc}[1]{\mathcal{#1}}
\newcommand{\bs}[1]{\boldsymbol{#1}}

\newcommand{\mcal}[1]{\mc{#1}}
\newcommand{\scp}[2]{\left<#1\,,\,#2\right>}

\newcommand{\ad}{\operatorname{ad}}

\def\p{{\partial}}

\def\ep{{\epsilon}}

\def\bR{{\mathbf{R}}}
\def\bu{{\boldsymbol{u}}}

\def\bx{{\mathbf{x}}}

\def\bxi{{\boldsymbol{\xi}}}

\def\p{\partial}

\pgfplotsset{compat=1.16}

\begin{document}

\title{Geometric theory of perturbation dynamics around non-equilibrium fluid flows}
\author{Darryl D. Holm, Ruiao Hu\footnote{Corresponding author, email: ruiao.hu15@imperial.ac.uk}, and Oliver D. Street\\
d.holm@imperial.ac.uk, ruiao.hu15@imperial.ac.uk, o.street18@imperial.ac.uk\\
Department of Mathematics, Imperial College London \\ SW7 2AZ, London, UK}
\date{\today}

\maketitle
\begin{abstract}
    The present work investigates the evolution of linear perturbations of time-dependent ideal fluid flows with advected quantities, expressed in terms of the second order variations of the action corresponding to a Lagrangian defined on a semidirect product space. This approach is related to Jacobi fields along geodesics and several examples are given explicitly to elucidate our approach. Numerical simulations of the perturbation dynamics are also presented. 
\end{abstract}

\tableofcontents

\section{Introduction}\label{Intro-sec}
We are dealing with the stability analysis of ideal fluid dynamics when the perturbations are in the form of displacement vector fields. In our approach, the displacement vector fields are shown to possess their own dynamics which depend functionally on the unperturbed fluid flow. The methodology we employ to derive the dynamics of the unperturbed flow and its perturbations is the Euler-Poincar\'e variational principle \cite{HMR1998} with higher-order variations. In the Euler-Poincar\'e variational principle, the 1$^{st}$ order variations yield essentially all of the well known models of ideal fluid dynamics \cite{HMR1998b}. When the stationary condition of the 1$^{st}$ order variations are satisfied, the 2$^{nd}$ order variations yield the dynamics of the perturbations of the steady flows arising from the 1$^{st}$ order variations. The advantage of this variational approach is that the dynamics of the perturbations arising from the 2$^{nd}$ order variations are linear perturbation equations for arbitrary time-dependent flows. Thus, it is a generalisation to the fluid equilibrium flows that are typically used in stability analysis.

    The emergence of fluid models from the 1$^{st}$ order variations in Hamilton's Principle was first revealed long before VI Arnold's observation \cite{Arnold1966} that the solutions of Euler's fluid equations represent time ($t$) dependent geodesic paths $g_t$  on  the manifold of smooth invertible maps (diffeomorphisms). That is, $g_t\in {\rm Diff}(M)$ acts on the fluid reference configuration $({\rm Diff}\times M\to M$) in the flow domain $M$, whose paths $\bx_t(\bx_0)=g_t\bx_0$ with $g_0\bx_0=\bx_0$ are the trajectories of Lagrangian fluid parcels. Arnold's observation opened the flood gates of new mathematical research in fluid dynamics. For a review, see e.g., \cite{AK1998}. In Arnold \cite{Arnold1966}, the variational principle which governs the Euler fluid motion was found to be the Hamilton principle $\delta S=0$ with $S=\int_0^T\ell(u)\,dt$ whose Lagrangian $\ell(u)$ is the fluid kinetic energy $\tfrac12\|u\|^2_{L^2}$ for fluid velocity, $u$. The fluid kinetic energy serves as the metric on the tangent space of the diffeomorphisms, expressed in terms of the \emph{Eulerian} velocity vector fields. The variations are taken with respect to the infinitesimal action of the diffeomorphisms on volume-preserving (spatial) vector fields. Since this observation, the framework was generalised to semidirect product spaces where the symmetry is broken by the inclusion of advected quantities \cite{HMR1998b}. As a result, the characterisation of a broad class of models through the variational approach became available. These developments related \emph{symmetry reduction} to fluid dynamics.

The 2$^{nd}$ order variations of the Euler-Poincar\'e variational principle have already been effective in deriving mean-flow equations \cite{Holm1999PhysD,Holm2002Chaos,Holm2002PhysD} as small-amplitude generalised Lagrangian mean (called $\mathfrak{glm}$) equations leading to turbulence models such as the Navier-Stokes-alpha model and its ideal version the Euler-alpha model. The alpha turbulence models introduced in \cite{Chen-etal1998,HMR1998b} were derived by applying the Lagrangian mean to the 2$^{nd}$ order variations for fluid dynamics in \cite{Holm1999PhysD,Holm2002Chaos,Holm2002PhysD}. They were then analysed mathematically in \cite{FHT2001,FHT2002} and applied computationally to primitive-equation global ocean circulation models in \cite{Hecht-etal2008a,Hecht-etal2008b,Hecht-etal2008c}. Furthermore, a relationship exists between the 2$^{nd}$ order variations of the Euler-Poincar\'e variational principle, Jacobi fields and the Jacobi equations. The Jacobi field equations for the evolution of an initial Lagrangian displacement away from a geodesic flow are usually treated in terms of covariant derivatives. However, here we will discuss Jacobi fields in the Euler-Poincar\'e framework. For examples and references to the classical approach to the treatment of Jacobi fields, see \cite{CF1996, Chiaffredo-etal2023,JJost2023,Jost2008,Michor2006,ModinPerrot2023,Preston2004,WashabaughPreston2017,Younes2007}.

Stability analysis is a vast field with a long history. Below, we will briefly review the Jacobi approach to stability analysis to better illuminate how this paper contributes to the literature. 

\paragraph{Brief review of the Jacobi approach to stability analysis.}
Stationary variational principles deal with balance and closure in dynamical systems. 
The dynamics near balance and the discovery of imbalance is the province of perturbation theory.
Higher order variational principles which govern the perturbations can tell us about instability and the initial phases of imbalance. 
Of course, for nonlinear systems tipping points also may exist and those can take the system to states far away from balance.
For viscous fluids, the primary tipping point is the onset of turbulence, whose true nature remains fascinating but elusive and beyond the treatment of linear imbalance and instability of ideal fluid flow considered here.

The classical studies of imbalance and instability refer to the dynamical behaviour of solutions near equilibria. One of the most beautiful mathematical theories of imbalance was introduced by Jacobi to describe solutions near geodesic flows. Besides Jacobi, this topic also stimulated research by the likes of Dirichlet, Dedekind, Riemann, Poincar\'e, and Lyapunov \cite{Sreenivasan2019}. Jacobi's theory and Riemann's results inspired Chandrasekhar's focus on ellipsoidal figures of equilibrium of rotating self-gravitating fluids \cite{Chandra1969}. That focus on steady ellipsoidal fluid configurations in turn led Chandrasekhar to his work on the emission of gravity waves by rotating ellipsoidal fluid masses in \cite{Chandra1970}, and eventually to Chandrasekhar's mathematical theory of black holes in \cite{Chandra1983}. The angular momentum of the rotating fluid body is the source of the emission of gravity waves. However, in the current paper we shall be concerned with the dual of angular momentum
introduced by Dedekind: namely, we shall focus on the fluid circulation in a fixed frame.


In Arnold \cite{Arnold1965}, a stationary flow of an ideal fluid is shown to be Lyapunov stable if the quadratic form given by the second variation of the kinetic energy restricted to coadjoint orbits in the algebra of smooth divergence-free vector fields is either positive, or sufficiently negative. The Hamiltonian version of Arnold's stability result for the Euler fluid equations was obtained by applying the Legendre transformation to their application of Hamilton's principle. This step led to new methods for determining sufficient conditions for nonlinear stability of fluid and plasma equilibria which include potential energy such as thermodynamics and magnetic energy, as well as kinetic energy; see, e.g., \cite{HMRW1985}. 
The present work investigates the quadratic form given by the second variation of the Hamilton principle whose fluid Lagrangian contains both the kinetic and potential energy. In addition, this work investigates the second-variation Hamilton's principle for time-dependent fluid flows, and is not limited to time-independent fluid equilibria.

\paragraph{Plan of the paper.} Section \ref{sec2-symred} sets the stage for the remainder of the paper, by defining symmetry-reduced variational principles at 1$^{st}$ and 2$^{nd}$ order. By considering geodesics of a right-invariant metric on a Lie group, the equations resulting from the 2$^{nd}$ variation are shown to be an extension of the literature on Jacobi fields. Section \ref{sec3-examples} presents several examples of linearisation of well-known Euler-Poincar\'e fluid equations based on second-order symmetry-reduced variational principles. The examples in Section \ref{sec3-examples} demonstrate that the current approach is not limited to geodesic motions, nor is it limited to steady flows. Section \ref{sec:numerics} presents numerical simulations for the Euler-Bousinessq equations and their perturbation equations resulting from the $2^{nd}$ order variations in a vertical slice domain. In Section \ref{sec4-sumout}, we summarise the present results and discuss future developments. 

\section{Euler-Poincar\'e variational principles and their linearisation}\label{sec2-symred}

Following Arnold's identification of fluid flows as paths $g_t$ on the manifold of smooth invertible maps, it is natural to consider the dynamics of imbalance induced by perturbations of nonequilibrium fluid flows in the light of the Jacobi equations for geodesic flows \cite{Jost2008,Michor2006}. Here, we are concerned with time-parameterised curves $g_t$ in the diffeomorphism group, and a family of variations of these curves, parameterised by $\epsilon\in\mathbb{R}$, $g_{\epsilon,t}$ such that $g_{0,t} = g_t$. We seek to extend the Jacobi equations of displacement dynamics near a geodesic flow to apply in semidirect product spaces which admit the dynamics of an Eulerian vector field $\xi(\bx,t)\in\mathfrak{X}(M)$ representing the displacement defined by
\begin{align}
\xi(\bx_t,t) :=  \frac{\p g_{t,\ep}}{\p \ep}\Big|_{\ep=0}g_t^{-1}\bx_t = \delta g_t \bx_0 =: \delta g_t g_t^{-1}\bx_t \,,
\label{ep-def}
\end{align}
for each fluid element in the flow $\bx_t=g_t\bx_0$ initially at the reference position $g_0\bx_0=\bx_0$, where $g_{t,\ep}$ are arbitrary disturbances of $g_t$ near the identity of the group of diffeomorphisms. This infinitesimal displacement vector field appears naturally in the Euler-Poincar\'e variational principle which we will discuss next before considering the correspondence to Jacobi fields and the Jacobi equation. 

\subsection{The Euler-Poincar\'e and Lie-Poisson equations} \label{subsec:EP_and_LP}
In the Euler-Poincar\'e theory of ideal fluid dynamics \cite{HMR1998}, the fluid is formally described by the elements in the semidirect product space comprising the vector fields $\mathfrak{X}(M)$, and the space of advected quantities, $V^*$, on which there exists a right representation of ${\rm Diff}(M)$ by pullback. The space of vector fields $\mathfrak{X}(M)$ contains the fluid's Eulerian velocity vector field $u$, defined in terms of curves, $g_t$, in the diffeomorphism group as $ u :=\dot{g}_t g_t^{-1}$.
The vector space $V^*(M)$, defined following convention as the dual to a vector space $V(M)$, is the space of advected quantities. Let $a\in V^*(M)$, we say that $a$ is an advected quantity if it satisfies the pushforward relation $a_t = a_0 g_t^{-1}$. This $a_t$ is the global solution of the advection relation,
\begin{align}
\p_t a_t = - \mathcal{L}_{u_t} a_t \,,
\label{advect-def}
\end{align}
in which $\mathcal{L}_{u_t} $ denotes Lie derivative with respect to the time-dependent fluid velocity vector field $u_t\in\mathfrak{X}(M)$.
Advected quantities are tensor fields, examples of which are mass density (a volume form) or potential temperature (a $0$-form). Arising from the variations of the paths $g_t$ on the manifold of diffeomorphisms, the corresponding variations of the fluid velocity $u = \dot{g}_tg_t^{-1}$ and the advected quantity $a_t  = a_0 g_t^{-1}$ are given, respectively, by  \cite{HMR1998}
\begin{align}
\delta u = \p_t \xi - {\rm ad}_u \xi := \p_t \xi + \big[ u \,,\,\xi \big]
\quad\hbox{and}\quad
\delta a = - \mathcal{L}_{\xi} a
\,.
\label{EP-var-def}
\end{align}
Here, the displacement vector field $\xi\in\mathfrak{X}(M)$ is defined in equation \eqref{ep-def} and $- \,{\rm ad}_u \xi := \big[ u \,,\,\xi \big]\in \mathfrak{X}(M)$ for right adjoint Lie algebra action is the Jacobi-Lie bracket of the vector fields $u$ and $\xi$. 

Let $\ell(u,a)$ denote the fluid Lagrangian. By applying Hamilton's principle, $\delta S = 0$, to the action integral $S=\int_0^T \ell(u,a)\,dt$ with the constrained variations of fluid variables $(u, a)$ as given in \eqref{EP-var-def}, we find
\begin{align}
\begin{split}
0 = \delta S &= \delta \int_0^T \ell(u,a)\,dt
\\&= 
\int_0^T \scp{\frac{\delta\ell }{\delta u} }{{\delta u}} +  \scp{\frac{\delta\ell }{\delta a} }{{\delta a}}
\,dt
\\&= 
\int_0^T \scp{\frac{\delta\ell }{\delta u} }{\p_t \xi - {\rm ad}_u \xi  } +  \scp{\frac{\delta\ell }{\delta a} }{- \mathcal{L}_{\xi} a}
\,dt
\\&= 
\int_0^T \scp{- \big( \p_t  + {\rm ad}^*_u \big)\frac{\delta\ell }{\delta u} + \frac{\delta\ell }{\delta a} \diamond a }{\xi  } \,dt
+ \scp{\frac{\delta\ell }{\delta u}}{\xi}\Big|_0^T
\,,\end{split}
\label{HP4EP}
\end{align}
where the brackets $\scp{\cdot}{\cdot}$ denote the $L^2$ pairing on
the flow manifold $M$ and the diamond $(\diamond)$ operator is defined as
\begin{align}
\scp{\frac{\delta\ell }{\delta a} }{- \mathcal{L}_{\xi} a}_{V\times V^*} =: \scp{ \frac{\delta\ell }{\delta a} \diamond a }{\xi}_{\mathfrak{X}^*(M)\times \mathfrak{X}(M)} 
\label{diamond-def}
\end{align}
in which we have applied natural boundary conditions when integrating by parts in space. The stationarity condition in Hamilton's principle $\delta S = 0$ with vanishing endpoint conditions in time on $\xi = \delta g\cdot g^{-1}$ then yields the Euler-Poincar\'e equation of fluid motion
\begin{equation}
\big(\p_t + \mathcal{L}_{u} \big) \frac{\delta\ell }{\delta u} = \frac{\delta\ell }{\delta a} \diamond a
\,,
\label{EP-eqns1}
\end{equation}
where the advected quantities $a$ satisfy the advection relation,
\begin{equation}
    \big(\p_t + \mathcal{L}_{u} \big) a = 0 \,.
\end{equation}
Note that the fact that $a$ is an advected quantity is encoded within the variational procedure \eqref{HP4EP} within the form of the variation $\delta a$.
This calculation leads us to the Kelvin-Noether theorem, written below.

\begin{theorem}[Kelvin-Noether theorem \cite{HMR1998}]\label{KN-thm}
Given the local advection of mass by fluid transport, 
\begin{align}
\big(\p_t   + \mathcal{L}_{u} \big) (Dd^nx) = 0
\,,\label{D-eqn}
\end{align}
implied by the push-forward relation $D_td^nx_t = g_{t\,*}(D_0d^nx_0)$ for each fluid element in the flow $\bx_t=g_t\bx_0$ initially at the reference position $g_0\bx_0=\bx_0$ and volume element $d^nx$ in $n$ dimensions, then the Euler-Poincar\'e equation of fluid motion \eqref{EP-eqns1} implies the Kelvin-Noether relation,
\begin{align}
\frac{d}{dt}\oint_{C(u)} \frac{1}{D}\frac{\delta\ell }{\delta u}
=
\oint_{C(u)}\frac{1}{D}\frac{\delta\ell }{\delta a} \diamond a
\,,\label{KN-eqn-1st}
\end{align}
for any material loop $C(u)$ moving with the flow velocity $u=\dot{g}_tg_t^{-1}$.
\end{theorem}
The Euler-Poincar\'e equation \eqref{EP-eqns1}, together with the advection relation for $a$, can be shown to be equivalent to semidirect-product Lie-Poisson equations on $\mathfrak{X}^*\ltimes V^*$, where $\ltimes$ denotes the semidirect product. Indeed, we make the following Legendre transformation from $\mathfrak{X}$ to $\mathfrak{X}^*$
\begin{equation*}
    m = \frac{\delta\ell}{\delta u} \,,\quad\hbox{and}\quad h(m,a) = \scp{m}{u} - \ell(u,a) \,.
\end{equation*}
Under the assumption that the map $u \rightarrow m$ is a diffeomorphism from $\mathfrak{X}$ to $\mathfrak{X}^*$, we have that $u = \delta h / \delta m$ and the Lie-Poisson equations arise from the following application of Hamilton's principle.
\begin{align*}
\begin{split}
0 = \delta S &= \delta \int_0^T \ell(u,a)\,dt = \delta \int_0^T \scp{m}{u} - h(m,a)\,dt
\\&= 
\int_0^T \scp{m}{\delta u} + \scp{\delta m}{u - \frac{\delta h}{\delta m}} - \scp{\frac{\delta h }{\delta a} }{{\delta a}}
\,dt
\\&= 
\int_0^T \scp{m}{\p_t \xi - {\rm ad}_u \xi  }+ \scp{\delta m}{u - \frac{\delta h}{\delta m}} +  \scp{\frac{\delta h}{\delta a} }{\mathcal{L}_{\xi} a}
\,dt
\\&= 
\int_0^T \scp{- \big( \p_t  + {\rm ad}^*_u \big)m - \frac{\delta h }{\delta a} \diamond a }{\xi  }+ \scp{\delta m}{u - \frac{\delta h}{\delta m}} \,dt
+ \scp{\frac{\delta\ell }{\delta u}}{\xi}\Big|_0^T
\,.\end{split}
\label{HP4LP}
\end{align*}
This calculation yields the implicit form of the following Lie-Poisson equation for fluid motion
\begin{equation}\label{LP-eqns1}
\big(\p_t   + \mathcal{L}_{\delta h / \delta m} \big) m = -\frac{\delta h }{\delta a} \diamond a
\,,\quad\hbox{and}\quad
\big(\p_t   + \mathcal{L}_{\delta h / \delta m} \big) a = 0
\,.
\end{equation}
Note that these equations are equivalent to the Lie-Poisson equation
\begin{equation}
    (\p_t + \ad^*_{\delta h / \delta \mu})\mu = 0 \,,\quad\hbox{for}\quad \mu = (m,a) \in \mathfrak{X}^*\ltimes V^* \,,
\end{equation}
where $\ad^*$ is the coadjoint representation of $\mathfrak{X}\ltimes V$ acting on 
its dual $\mathfrak{X}^*\ltimes V^*$.

\subsection{The second variation.}
The first and second variations are defined as
\begin{equation}
    \delta f(\mu;\delta\mu) := \frac{d}{d\epsilon}\bigg|_{\epsilon=0} f(\mu+\epsilon\delta\mu) \,,\quad\hbox{and}\quad \delta^2f(\mu;\delta\mu):= \frac{d^2}{d\epsilon^2}\bigg|_{\epsilon=0}f(\mu+\epsilon\delta\mu) \,,
\end{equation}
respectively. In each example we consider in this article, the second variation produces a symmetric bilinear form, which we will denote also by $\delta^2f(\delta\mu,\delta\mu)$. Recall the definition of a functional derivative,
\begin{equation*}
	\scp{\frac{\delta f}{\delta\mu}}{\delta\mu} := \delta f(\mu;\delta \mu) \,.
\end{equation*}
That is, $\scp{\delta f / \delta \mu}{\delta \mu}$ is the first term of the expansion around $\epsilon=0$ of the first derivative of $f(\mu+\epsilon\delta\mu)$ in $\epsilon$. When taking second variations, we will often wish to go further in this expansion. Indeed, we see that
\begin{equation}\label{eqn:useful_expansion}
\begin{aligned}
    \frac{d}{d\epsilon}f(\mu+\epsilon\delta\mu) &= \frac{d}{d\epsilon}\bigg|_{\epsilon=0} f(\mu+\epsilon\delta\mu) + \epsilon \frac{d^2}{d\epsilon^2}\bigg|_{\epsilon=0}f(\mu+\epsilon\delta\mu) + \mathcal{O}(\epsilon^2)
    \\
    &= \scp{\frac{\delta f}{\delta\mu}}{\delta \mu} + \epsilon \delta^2f(\delta\mu,\delta\mu) + \mathcal{O}(\epsilon^2)
    \\
    &= \scp{\frac{\delta f}{\delta\mu}}{\delta \mu} + \frac{\epsilon}{2}\scp{\frac{\delta(\delta^2f)}{\delta(\delta\mu)}}{\delta\mu}  + \mathcal{O}(\epsilon^2) \,,
\end{aligned}
\end{equation}
where the final equality holds since $\delta^2f(\delta\mu,\delta\mu)$ is a symmetric bilinear form. This calculation allows us to take functional derivatives to the next order of the expansion, since the term of order $\epsilon$ involves a pairing of an element of $\mathfrak{g}$ against $\delta\mu$.

\subsection{Linearised Euler-Poincar\'e and Lie-Poisson equations}
We may consider an expansion of the Lie-Poisson equation
\begin{equation}\label{eqn:LP}
	\p_t \mu + \ad^*_{\frac{\delta h}{\delta \mu}} \mu = 0 \,,
\end{equation}
by exploiting the calculation in equation \eqref{eqn:useful_expansion} to notice that
\begin{equation}
    \left(\p_t + \ad^*_{\frac{\delta h}{\delta \mu} + \frac{\epsilon}{2}\frac{\delta(\delta^2h)}{\delta(\delta\mu)} } \right)(\mu + \epsilon\delta\mu) = \mathcal{O}(\epsilon^3) \,.
\end{equation}
Notice that the first order terms in this equation correspond exactly to the Lie-Poisson equation \eqref{eqn:LP}, and the next order terms give the linearised equation for the perturbation $\delta \mu$ around the solution $\mu$ to the Lie-Poisson system
\begin{equation}\label{eqn:linearised_LP}
	\p_t\delta \mu + \ad^*_{\frac{\delta h}{\delta \mu}} \delta\mu = -\frac12 \ad^*_{\frac{\delta(\delta^2h)}{\delta(\delta\mu)}} \mu \,. 
\end{equation}
This can be written in terms of Poisson bracket-like objects as
\begin{equation}\label{eqn:linearised_LP_brackets}
	\p_t f(\delta \mu) = \scp{\delta \mu}{\left[ \frac{\delta f}{\delta(\delta \mu)} , \frac{\delta h}{\delta \mu} \right]} + \frac12 \scp{\mu}{\left[ \frac{\delta f}{\delta(\delta \mu)} , \frac{\delta(\delta^2 h)}{\delta(\delta \mu)} \right]} \,.
\end{equation}
For a fixed solution, $\mu$, of the Lie-Poisson equation, the second bracket here is a \emph{frozen} Lie-Poisson bracket. Notice that there is a direct connection here to the motion considered in a previous study \cite{HMRW1985}, where it was observed that for an equilibrium solution, $\mu_e$, corresponding to a critical point of $h+C$ for some Casimir, $C$, this equation for $\delta \mu$ is Hamiltonian with respect to the frozen Lie-Poisson bracket with Hamiltonian $\delta^2h_C$. The full equation \eqref{eqn:linearised_LP} considered here is \emph{not} itself a Lie-Poisson system.

\paragraph{Hamiltonian systems on semidirect product spaces and continuum dynamics.} When the configuration space is a semidirect product Lie group, $G\ltimes V$, where $G$ acts on $V$ through a right representation, the equations of motion on the Lie co-algebra, $\mathfrak{g}^*\ltimes V^*$, are
\begin{equation}\label{eqn:semidirectproduct_linearised_LP}
    \p_t (\delta\mu,\delta a) + \ad^*_{\frac{\delta h}{\delta(\mu,a)}}(\delta\mu,\delta a) = -\frac12 \ad^*_{\frac{\delta(\delta^2 h)}{\delta(\delta\mu,\delta a)}}(\mu,a) \,.
\end{equation}
These equations can be manipulated into a more convenient form by determining $\ad^*$ for the semidirect product space in the usual way. That is, an alternative form of equation \eqref{eqn:semidirectproduct_linearised_LP} can be directly deduced from equation \eqref{eqn:linearised_LP_brackets} by inserting the standard Lie bracket for semidirect product spaces as
\begin{equation}
\begin{aligned}
	\p_t f(\delta \mu,\delta a) &= \scp{(\delta \mu , \delta a)}{\left[ \frac{\delta f}{\delta(\delta \mu,\delta a)} , \frac{\delta h}{\delta (\mu,a)} \right]} + \frac12 \scp{(\mu,a)}{\left[ \frac{\delta f}{\delta(\delta (\mu,a))} , \frac{\delta(\delta^2 h)}{\delta(\delta \mu,\delta a)} \right]}
    \\
    &= \scp{\delta\mu}{\ad_{\frac{\delta f}{\delta(\delta\mu)}}\frac{\delta h}{\delta\mu}} + \scp{\delta a}{\mathcal{L}^T_{\frac{\delta f}{\delta(\delta\mu)}}\frac{\delta h}{\delta a} - \mathcal{L}^T_{\frac{\delta h}{\delta\mu}}\frac{\delta f}{\delta(\delta a)}} 
    \\
    &\qquad + \frac12\scp{\mu}{\ad_{\frac{\delta f}{\delta(\delta\mu)}}\frac{\delta(\delta^2h)}{\delta(\delta\mu)}} + \frac12\scp{a}{\mathcal{L}^T_{\frac{\delta f}{\delta(\delta\mu}}\frac{\delta(\delta^2h)}{\delta(\delta a)} - \mathcal{L}^T_{\frac{\delta(\delta^2h)}{\delta(\delta\mu)}}\frac{\delta f}{\delta(\delta a)}}\,,
\end{aligned}
\end{equation}
where $\mathcal{L}^T$ is the transpose of the Lie derivative. Integrating by parts gives the following equations for the dynamics of the perturbations,
\begin{align}
    \left( \p_t + \ad^*_{\frac{\delta h}{\delta\mu}} \right)\delta\mu &= -\frac12\ad^*_{\frac{\delta(\delta^2h)}{\delta(\delta\mu)}}\mu - \frac{\delta h}{\delta a} \diamond \delta a - \frac12 \frac{\delta(\delta^2 h)}{\delta(\delta a)}\diamond a
    \label{eqn:linearised_LP_semidirect1}
    \,,\\
    \left( \p_t + \mathcal{L}_{\frac{\delta h}{\delta\mu}} \right)\delta a &= - \frac12\mathcal{L}_{\frac{\delta(\delta^2 h)}{\delta(\delta \mu)}}a
    \,.
    \label{eqn:linearised_LP_semidirect2}
\end{align}
These comprise a useful form in which to write the equations, since the operator $\p_t + \ad^*_{\delta h / \delta \mu}$ is the standard geometric form of the advective derivative and mirrors the left hand side of the regular Lie-Poisson and Euler-Poincar\'e equations. It then remains only to determine the right hand sides of the equation in an analogous fashion to how fluid models are derived from the Euler-Poincar\'e equation \eqref{EP-eqns1}, as discussed in \cite{HMR1998}. It is well known that the equations with advected quantities which break the relabelling symmetry under the entire diffeomorphism group, whilst not Euler-Poincar\'e equations on the semidirect product space $\mathfrak{g}\ltimes V^*$, are the standard Lie-Poisson equations on the dual space $\mathfrak{g}^*\ltimes V^*$. Here, we have demonstrated that the same is true for the linearised equation on the semidirect product space. That is, the linearised system of equations \eqref{eqn:linearised_LP_semidirect1} and \eqref{eqn:linearised_LP_semidirect2} has the same geometric form as the linearised Lie-Poisson equation \eqref{eqn:linearised_LP}.

\paragraph{Euler-Poincar\'e equations for semidirect product spaces.} When the Legendre transform is well defined, the equations \eqref{eqn:linearised_LP_semidirect1} and \eqref{eqn:linearised_LP_semidirect2} have equivalent forms in terms of the Lagrangian, $\ell:\mathfrak{g}\times V^* \rightarrow \mathbb{R}$. As was achieved for the Lie-Poisson equations, one may deduce these equations directly from the following Euler-Poincar\'e equations
\begin{align}
    \left( \p_t + \ad^*_u \right)\frac{\delta \ell}{\delta u} &= \frac{\delta\ell}{\delta a}\diamond a
    \label{eqn:EP1}
    \,,\\
    \left( \p_t + \mathcal{L}_u \right)a &= 0
    \label{eqn:EP2}
    \,.
\end{align}
Again utilising the calculation performed in equation \eqref{eqn:useful_expansion}, we find the expansion of these equations as
\begin{align*}
    \left( \p_t + \ad^*_{u+\epsilon\delta u} \right)\left(\frac{\delta \ell}{\delta u} + \frac{\epsilon}{2}\frac{\delta(\delta^2\ell)}{\delta(\delta u)}\right) &= \left(\frac{\delta \ell}{\delta a} + \frac{\epsilon}{2}\frac{\delta(\delta^2\ell)}{\delta(\delta a)}\right)\diamond \left(a + \epsilon\delta a \right) + \mathcal{O}(\epsilon^3)
    \,,\\
    \left( \p_t + \mathcal{L}_{u+\epsilon\delta u} \right)\left( a + \epsilon\delta a \right)&= \mathcal{O}(\epsilon^3)
    \,.
\end{align*}
As for the Hamiltonian case, the first order terms are the Euler-Poincar\'e equations and the linearised equations are given by the order $\epsilon$ terms as
\begin{align}
    \left( \p_t + \ad^*_u \right)\left( \frac12\frac{\delta(\delta^2\ell)}{\delta(\delta u)} \right) &= -\ad^*_{\delta u}\frac{\delta\ell}{\delta u} + \frac{\delta \ell}{\delta a} \diamond \delta a + \frac12 \frac{\delta(\delta^2 \ell)}{\delta(\delta a)}\diamond a
    \label{eqn:linearised_EP_semidirect1}
    \,,\\
    \left( \p_t + \mathcal{L}_{u} \right)\delta a + \mathcal{L}_{\delta u}a &= 0
    \,.
    \label{eqn:linearised_EP_semidirect2}
\end{align}

\subsection{Second order variations of the Euler-Poincar\'e variational principle}\label{subsec:variational_linearisation}

In the previous linearised equations, the perturbations $(\delta u , \delta a)$ have been arbitrary. When understanding such equations \emph{within} the variational principle itself, these perturbations become constrained.

\paragraph{The Euler-Poincar\'e variational principle.} The equations \eqref{eqn:linearised_EP_semidirect1} and \eqref{eqn:linearised_EP_semidirect2} can also be deduced by considering the next variation in Hamilton's principle. When deducing symmetry reduced equations from the variational principle, arbitrary variations in the group are not arbitrary in the algebra. In particular, for an Eulerian velocity vector field $u = \p_t g \cdot g^{-1} \in \mathfrak{g}$, where concatination denotes the lifted right translation of $\p_tg\in T_gG$ by $g^{-1}$, an arbitrary variation $\delta g$ of the group element allows the variation of $u$ to be expressed in terms of an arbitrary vector field $\xi = \delta g \cdot g^{-1} \in \mathfrak{g}$. We may deduce the forms of the first and second variation by expanding the vector field 
$u_\ep=\p_tg_{t,\ep}\cdot g_{t,\ep}^{-1}$ in a Taylor series in powers of a small parameter $\ep\ll1$ around the identity $\ep=0$ as
\begin{align}
u_\ep &= u + \ep (\p_t \xi - {\rm ad}_u \xi) + \frac{\ep^2}{2} \left(\p_t \delta\xi - {\rm ad}_{\delta u} \xi - {\rm ad}_u \delta\xi \right) + \mathcal{O}(\ep^2)
=: u + \ep\delta u +  \frac{\ep}{2}\delta^2 u + \mathcal{O}(\ep^2)
\,.\label{Var-exp-def-}
\end{align}
Likewise, for an advected quantity which evolves by push-forward as $a_t = a_0g_t^{-1}$ one defines the Taylor series for the variation as
\begin{align}
a_\ep &= u|_{\ep=0} - \ep \mathcal{L}_\xi a - \frac{\ep^2}{2}\left( \mathcal{L}_{\delta\xi} a + \mathcal{L}_\xi \delta a\right) +\mathcal{O}(\ep^2)
=: a +  \ep\delta a +  \frac{\ep}{2}\delta^2 a +\mathcal{O}(\ep^2)
\,.\label{Var-exp-def-a}
\end{align}
Equations \eqref{Var-exp-def-} and \eqref{Var-exp-def-a} comprise a second order extension of the \emph{Lin constraints} used to derive the original Euler-Poincar\'e equations (see e.g. \cite{BKMR1996}).

As in Section \ref{subsec:EP_and_LP}, the Euler-Poincar\'e equations \eqref{eqn:EP1} and \eqref{eqn:EP2} can be deduced from Hamilton's principle by making use of these variations 
\begin{equation}\label{eqn:HamPrin_EP}
\begin{aligned}
    0 &= \delta S = \delta\int \ell(u,a)\,dt = \int \scp{\frac{\delta\ell}{\delta u}}{\delta u} + \scp{\frac{\delta\ell}{\delta a}}{\delta a} \,dt
    \\
    &= \int \scp{\frac{\delta\ell}{\delta u}}{\p_t \xi - {\rm ad}_u \xi} - \scp{\frac{\delta\ell}{\delta a}}{\mathcal{L}_\xi a} \,dt = -\int \scp{\left( \p_t + \ad^*_u \right)\frac{\delta\ell}{\delta u} - \frac{\delta\ell}{\delta a}\diamond a}{\xi} \,dt \,.
\end{aligned}
\end{equation}
At the second order, we have
\begin{equation}
\begin{aligned}
    0=\delta^2S &= \int \scp{\frac12\frac{\delta(\delta^2\ell)}{\delta(\delta u)}}{\delta u} + \scp{\frac12\frac{\delta(\delta^2\ell)}{\delta(\delta a)}}{\delta a} + \scp{\frac{\delta\ell}{\delta u}}{\delta^2 u} + \scp{\frac{\delta\ell}{\delta a}}{\delta^2 a} \,dt
    \\
    &= \int \scp{\frac12\frac{\delta(\delta^2\ell)}{\delta(\delta u)}}{\p_t \xi - {\rm ad}_u \xi} - \scp{\frac12\frac{\delta(\delta^2\ell)}{\delta(\delta a)}}{\mathcal{L}_\xi a}
    \\
    &\qquad\qquad + \scp{\frac{\delta\ell}{\delta u}}{\p_t \delta\xi - {\rm ad}_{\delta u} \xi - {\rm ad}_u \delta\xi} - \scp{\frac{\delta\ell}{\delta a}}{ \mathcal{L}_{\delta\xi} a + \mathcal{L}_\xi \delta a} \,dt
    \\
    &= \int \scp{-(\p_t + \ad^*_u)\frac{\delta\ell}{\delta u}+\frac{\delta\ell}{\delta a}\diamond a}{\delta \xi}
    \\
    &\qquad\qquad + \scp{-(\p_t +\ad^*_u)\frac12\frac{\delta(\delta^2\ell)}{\delta(\delta u)} + \frac12\frac{\delta(\delta^2\ell)}{\delta(\delta a)}\diamond a + \frac{\delta\ell}{\delta a}\diamond \delta a - \ad^*_{\delta u}\frac{\delta\ell}{\delta u}}{\xi}\,dt \,.
\end{aligned}
\end{equation}
The arbitrary nature of the vector field $\xi$ and, by extension, $\delta\xi$ gives the Euler-Poincar\'e equation \eqref{eqn:EP1} and the equation \eqref{eqn:linearised_EP_semidirect1} for the linear perturbation. This is the symmetry-reduced version of taking second order variations in first order variational principle used in e.g. \cite{CF1996} to derive dynamics of Jacobi fields.

In practice, this process will give us equations for $\delta u$ and $\delta a$. Since $(\delta u , \delta a)$ can be expressed in terms of the arbitrary variable $\xi = \delta g \cdot g^{-1}$, these imply an equation for $\xi$. However, as the following Proposition demonstrates, deriving an equation for $\xi$ only requires the equation for $\delta u$, since the equation for $\delta a$ is satisfied trivially.
\begin{proposition}
    Given the constrained form of the first variations
    \begin{equation}\label{eqn:Lin_constraints}
        \delta u = \p_t\xi -\ad_u\xi
        \,,\quad\hbox{and}\quad
        \delta a = -\mathcal{L}_ua
        \,,
    \end{equation}
    and the fact that $a\in V^*$ is advected by $u\in\mathfrak{X}(M)$, we find that the equation \eqref{eqn:linearised_EP_semidirect2} is satisfied.
\end{proposition}
\begin{proof}
    We may verify this claim by direct computation. Indeed, notice that
    \begin{equation*}
    \begin{aligned}
        \p_t\delta a + \mathcal{L}_u \delta a + \mathcal{L}_{\delta u}a &= -\p_t\mathcal{L}_\xi a - \mathcal{L}_a\mathcal{L}_\xi a + \mathcal{L}_{\p_t\xi - \ad_u\xi}a
        \\
        &= -\mathcal{L}_{\p_t\xi}a - \mathcal{L}_\xi \p_ta - \mathcal{L}_a\mathcal{L}_\xi a + \mathcal{L}_{\p_t\xi}a - \mathcal{L}_{\ad_u\xi}a
        \\
        &= \mathcal{L}_\xi\mathcal{L}_ua - \mathcal{L}_u\mathcal{L}_\xi a - \mathcal{L}_{\ad_u\xi}a = 0 \,,
    \end{aligned}
    \end{equation*}
    where the final line is a consequence of the standard relationship between the adjoint representation and the Lie bracket for right-invariant systems.
\end{proof}

\paragraph{The Legendre transform and Lie-Poisson equations.} As was described in Section \ref{subsec:EP_and_LP}, the Lie-Poisson equations on semidirect product Lie co-algebras can be deduced by Legendre transforming within the application of Hamilton's Principle illustrated by equation \eqref{eqn:HamPrin_EP}. That is, we consider the following variational problem
\begin{equation}\label{eqn:HamPrin_LP}
\begin{aligned}
    0 = \delta S(\mu,a;u) &= \delta\int \scp{\mu}{u}-h(\mu,a)\,dt = \int \scp{\delta\mu}{u - \frac{\delta h}{\delta\mu}}+\scp{\mu}{\delta u} - \scp{\frac{\delta h}{\delta a}}{\delta a} \,dt
    \\
    &= \int \scp{\delta\mu}{u - \frac{\delta h}{\delta\mu}}+\scp{\mu}{\p_t \xi - {\rm ad}_u \xi} + \scp{\frac{\delta h}{\delta a}}{\mathcal{L}_\xi a} \,dt \\
    &= \int \scp{\delta\mu}{u - \frac{\delta h}{\delta\mu}}-\scp{\left( \p_t + \ad^*_u \right)\mu + \frac{\delta h}{\delta a}\diamond a}{\xi} \,dt \,,
\end{aligned}
\end{equation}
which yields Lie-Poisson equation \eqref{LP-eqns1}. Again, considering the second variation of this we have
\begin{equation}
\begin{aligned}
    0=\delta^2S(\mu,a;u) &= \int \scp{\delta\mu}{\delta u - \frac12\frac{\delta(\delta^2h)}{\delta(\delta \mu)}} + \scp{\delta^2\mu}{u - \frac{\delta h}{\delta\mu}}
    \\
    &\qquad\qquad +\scp{\delta\mu}{\delta u} -\scp{\frac12\frac{\delta(\delta^2h)}{\delta(\delta a)}}{\delta a}+ \scp{\mu}{\delta^2u} - \scp{\frac{\delta h}{\delta a}}{\delta^2 a} \,dt
    \\
    &= \int \scp{\delta\mu}{\delta u - \frac12\frac{\delta(\delta^2h)}{\delta(\delta \mu)}} + \scp{\delta^2\mu}{u - \frac{\delta h}{\delta\mu}}
    \\
    &\qquad\qquad +\scp{\delta\mu}{\p_t \xi - \ad_u \xi} +\scp{\frac12\frac{\delta(\delta^2h)}{\delta(\delta a)}}{\mathcal{L}_{\xi} a}
    \\
    &\qquad\qquad + \scp{\mu}{\p_t \delta\xi - \ad_{\delta u} \xi - \ad_{u}\delta\xi} + \scp{\frac{\delta h}{\delta a}}{\mathcal{L}_{\xi}\delta a + \mathcal{L}_{\delta\xi}a} \,dt
    \\
    &= \int \scp{\delta\mu}{\delta u - \frac12\frac{\delta(\delta^2h)}{\delta(\delta \mu)}} + \scp{\delta^2\mu}{u - \frac{\delta h}{\delta\mu}}+ \scp{-(\p_t + \ad^*_u)\mu-\frac{\delta h}{\delta a}\diamond a}{\delta \xi}
    \\
    &\qquad\qquad + \scp{-(\p_t +\ad^*_u)\delta\mu - \frac12\frac{\delta(\delta^2h)}{\delta(\delta a)}\diamond a - \frac{\delta h}{\delta a}\diamond \delta a - \ad^*_{\delta u}\mu}{\xi}\,dt \,.
\end{aligned}
\end{equation}
The first two terms in the final line of this calculation give us the identities
\begin{equation*}
    u=\frac{\delta h}{\delta u}\,,\quad\hbox{and}\quad \delta u = \frac12\frac{\delta(\delta^2h)}{\delta(\delta \mu)}\,,
\end{equation*}
the arbitrary nature of $\delta\xi$ gives us the Lie-Poisson equation, and since $\xi$ is arbitrary we have the equation \eqref{eqn:linearised_LP_semidirect1} for the linear perturbation expressed in terms of the Hamiltonian.

\subsection{Jacobi fields and geodesics of a right-invariant metric on a Lie group}\label{subsec:Jacobi_geodesics}

In this section, we will specialise to the case in which the Euler-Poincar\'e equation and its linearisation arise from a Lagrangian which is defined by a right-invariant inner product. Since we are in the Euler-Poincar\'e setting \cite{HMR1998} and are motivated by fluid dynamics, it is natural to derive these equations in terms of pull-backs by diffeomorphisms and Lie derivatives with respect to the smooth vector fields which generate the diffeomorphisms. However, the literature on Jacobi fields in the context of Lie groups is predominantly concerned with expressions on the group, $G$, or algebra, $\mathfrak{g}$, rather than its dual. As such, the Jacobi equation is most familiarly expressed in terms of geometric objects such as the covariant derivative and Riemannian curvature, as opposed to the adjoint representation, Lie derivative, and diamond operation we have thus far employed. Of course, the transformation from Lie derivatives to covariant derivatives can be performed using standard methods for Riemannian spaces, \cite{Jost2008,Michor2006,Younes2007}. 

In this section, we will relate these fields by showing that, for the Lagrangian defined by our right-invariant inner product, the resulting Euler-Poincar\'e equation on $\mathfrak{g}^*$ is the geodesic equation. This will be illustrated by writing the equation in its equivalent form on $\mathfrak{g}$. Furthermore, since the equation is then expressed on $\mathfrak{g}$, we will demonstrate that the linearised Euler-Poincar\'e equation is equivalent to the Jacobi equation in which the Jacobi field is defined by the arbitrary variations in the Euler-Poincar\'e variational principle. This will involve a shift in notation in this section relative to the others in this paper.
   
For a Lie group, $G$, there exists a natural duality pairing between its algebra $\mathfrak{g}$, and its co-algebra $\mathfrak{g}^*$, which we will denote by $\scp{\cdot}{\cdot}_{\mathfrak{g}^*\times\mathfrak{g}}:\mathfrak{g}^*\times\mathfrak{g}\rightarrow\mathbb{R}$. If the Lie group is augmented further with a right-invariant Riemannian metric, then the algebra possesses a (weak) right-invariant inner product, denoted by $\scp{\cdot}{\cdot}_{\mathfrak{g}}:\mathfrak{g}\times\mathfrak{g}\rightarrow\mathbb{R}$. This metric permits the discussion of \emph{geodesics} in this setting. Indeed, consider the following application of Hamilton's Principle
\begin{equation}\label{eqn:geodesic_dualspace}
    0 = \delta \int \frac12\scp{u^\flat}{u}_{\mathfrak{g}^*\times\mathfrak{g}}\,dt \,,
\end{equation}
where $\flat:\mathfrak{g}\rightarrow \mathfrak{g}^*$ is the musical mapping defined by $\scp{u^\flat}{v}_{\mathfrak{g}^*\times\mathfrak{g}} = \scp{u}{v}_{\mathfrak{g}}$ for $u,v\in \mathfrak{g}$. We note that when the inner product is weak, the musical mapping $\flat$ is not surjective \cite{M2015}.
Thus, the inverse musical mapping $\sharp:\mathfrak{g}^*\rightarrow\mathfrak{g}$ is defined, on the image of $\flat$, by $\scp{\alpha}{v}_{\mathfrak{g}^*\times\mathfrak{g}} = \scp{\alpha^\sharp}{v}_{\mathfrak{g}}$ for $\alpha\in\mathfrak{g}^*$ and $v\in\mathfrak{g}$. When the inner product is strong, the musical mappings become isomorphisms.
Following the calculations in Section \ref{subsec:variational_linearisation}, we see that the variational problem \eqref{eqn:geodesic_dualspace} implies the following Euler-Poincar\'e equation
\begin{equation}\label{eqn:geodesic_equation_dualspace}
    \p_tu^\flat + \ad^*_uu^\flat = 0 \,.
\end{equation}
The equation for $\delta u$, given by going to the second order in Hamilton's Principle, is
\begin{equation}\label{eqn:Jacobi_equation_dualspace}
    \p_t\delta u^\flat + \ad^*_u\delta u^\flat = -\ad^*_{\delta u}u^\flat \,,
\end{equation}
and, when considered together with the constrained variations \eqref{eqn:Lin_constraints}, implies the following equation,
\begin{equation}
    (\p_t + \ad^*_u)\big(\left( \p_t\xi - \ad_u\xi \right)^\flat\big) = -\ad^*_{\p_t\xi - \ad_u\xi}u^\flat \,,
\end{equation}
written in terms of the variable $\xi = \delta g \cdot g^{-1} \in \mathfrak{g}$, which is a \emph{Jacobi field}. We will now seek to formalise this notion.

In what follows, will use definitions and notation following Michor \cite{Michor2006} and readers should consult this text for additional details of this construction. Consider a family of time-parameterised geodesics, $\{ g_{t,s} \}_{s\in\mathbb{R}}$. Then we define $u,\xi \in \mathfrak{g}$ by $u = \left[(\p_tg_{t,s})\cdot g_{t,s}\right]_{s=0}$ and $\xi = \left[(\p_sg_{t,s})\cdot g_{t,s}\right]_{s=0}$, and these associations can be understood as smooth maps $C^{\infty}(G,\mathfrak{g})$. Furthermore, the lifted right action by $g_{t,s}$ provides a map from $C^{\infty}(G,\mathfrak{g})$ to the space of vector fields on $G$, and as such we have an isomorphism, $C^{\infty}(G,\mathfrak{g})\cong\mathfrak{X}(G)$. Notice that $\p_tg_{t,s},\p_sg_{t,s}\in T_{g_{t,s}}G$ and the pushforward of $\p_t$ and $\p_s$ by $g_{t,\epsilon}$ can be understood as vector fields in $G$. This permits us to introduce the notation $\nabla_{\p_t}\xi$ by which we mean the Levi-Civita covariant derivative of $\xi$ along the curve $g_t$, which can be interpreted as an element of $C^{\infty}(G,\mathfrak{g})$. Since we have constructed a covariant derivative in this manner, we may similarly define the Riemannian curvature as a function of $u$ and $\xi$, $\mathcal{R}(\xi,u):C^{\infty}(G,\mathfrak{g})\rightarrow C^{\infty}(G,\mathfrak{g})$. In particular, this is understood in the usual sense in terms of $\mathfrak{X}(G)$, where the isomorphism is applied at each stage to identify the vector fields on $G$ with elements of Lie algebra.

\begin{theorem}
    Suppose $u$ is a solution of the geodesic equation \eqref{eqn:geodesic_equation_dualspace}, then we have the following infinite dimensional analogue of the Jacobi equation
    \begin{equation}
        \nabla_{\p_t}\nabla_{\p_t}\xi + \mathcal{R}(\xi,u)u = 0 \,.
    \end{equation}
\end{theorem}
\begin{proof}
    Following on from the equations derived in Section \ref{subsec:variational_linearisation}, we have equations \eqref{eqn:geodesic_dualspace} and \eqref{eqn:Jacobi_equation_dualspace}. In order to prevent overuse of the musical isomorphisms, and to better connect with the existing literature in this direction, we will define $\ad^\dagger_{\Box}\Box:\mathfrak{g}\times\mathfrak{g}\rightarrow\mathfrak{g}$ in terms of $\ad^*_{\Box}\Box:\mathfrak{g}\times\mathfrak{g}^*\rightarrow\mathfrak{g}^*$ as
    \begin{equation}
        \ad^\dagger_u v = \left( \ad^*_u v^\flat \right)^\sharp \,.
    \end{equation}
    Note that this is simply the consequence of $\ad^*$ being the dual operator to $\ad$ with respect to the natural pairing $\scp{\cdot}{\cdot}_{\mathfrak{g}}$ and $\ad^\dagger$ being defined in the same manner with respect to the inner product $\scp{\cdot}{\cdot}_{\mathfrak{g}}$ on $\mathfrak{g}$. This allows us to write the equations entirely on the Lie algebra, without needing the dual space. Firstly, notice that applying Hamilton's Principle instead to the action $\frac12\int\scp{u}{u}_{\mathfrak{g}}\,dt$ yields the following equations
    \begin{align}
        \p_t u + \ad^\dagger_u u &= 0 \label{eqn:geodesic_algebra}
        \,,\\
        (\p_t+\ad^\dagger_u)(\p_t\xi - \ad_u\xi) &= - \ad^\dagger_{\p_t\xi - \ad_u\xi}u \label{eqn:Jacobi_equation_algebra}
        \,,
    \end{align}
    which correspond to rewriting equations \eqref{eqn:geodesic_dualspace} and \eqref{eqn:Jacobi_equation_dualspace} in terms of the operator $\ad^\dagger$. By using the linearity of $\ad^\dagger$ and moving terms in equation \eqref{eqn:Jacobi_equation_algebra} to the right hand side, we have
    \begin{align*}
        \p_{tt}\xi &= -\ad^\dagger_{\p_t\xi}u  + \ad_u\p_t\xi - \ad^\dagger_u\p_t\xi + \ad^\dagger_{\ad_u\xi}u+ \ad^\dagger_u\ad_u\xi+ \ad_{\p_tu}\xi
        \\
        \hbox{using equation \eqref{eqn:geodesic_algebra}}\quad &= -\ad^\dagger_{\p_t\xi}u  + \ad_u\p_t\xi - \ad^\dagger_u\p_t\xi+ \ad^\dagger_{\ad_u\xi}u + \ad^\dagger_u\ad_u\xi- \ad_{\ad^\dagger_uu}\xi
        \\
        \hbox{antisymmetry of the Lie bracket}\quad &= -\ad^\dagger_{\p_t\xi}u  + \ad_u\p_t\xi - \ad^\dagger_u\p_t\xi - \ad^\dagger_{\ad_\xi u}u - \ad^\dagger_u\ad_\xi u + \ad_\xi\ad^\dagger_uu
        \\
        &= -\ad^\dagger_{\p_t\xi}u  + \ad_u\p_t\xi - \ad^\dagger_u\p_t\xi + [ \ad^\dagger_\xi + \ad_\xi , \ad^\dagger_u ]u
        \,,
    \end{align*}
    where, in the final line, we have used the identity $-\ad^\dagger_{\ad_\xi u}u = [ \ad^\dagger_\xi , \ad^\dagger_u ]u$. As was shown by Michor \cite{Michor2006}, $\nabla_{\p_t}\nabla_{\p_t}\xi + \mathcal{R}(\xi,u)u$ is zero when the above equation is true.
\end{proof}

\section{Examples}\label{sec3-examples}

In each example, we will first derive the linearised equations in terms of arbitrary perturbations of our variables $(\delta u, \delta a)$. Following this, we will illustrate the equation in terms of the arbitrary vector field $\xi = \delta g \cdot g^{-1}$.

\subsection{The incompressible Euler equations}\label{subsec:Euler}

The $n$-dimensional Euler equations on a manifold $M$ correspond to a Lagrangian, $\ell_E:\mathfrak{X}(M)\times {\rm Den}(M)\rightarrow \mathbb{R}$, defined by
\begin{align}
    \ell_E(u,Dd^nx;\pi) &= \int_M \frac{D}{2}|\bu|^2 - \pi(D-1)\,d^nx 
    \label{eqn:Euler_Lagrangian}\,,\\
    \quad\hbox{and thus}\quad \frac12\delta^2\ell_E &= \int_M\frac{D}{2}|\bs{\delta u}|^2 + \delta D \bu \cdot \bs{\delta u} - \delta\pi \delta D\,d^nx \label{eqn:Euler_Lagrangian_var}\,,
\end{align}
where $u=\bu\cdot\nabla$ and $\delta u = \bs{\delta u}\cdot\nabla$ denote the vector fields expressed in terms of a basis. From equations \eqref{eqn:Euler_Lagrangian} and \eqref{eqn:Euler_Lagrangian_var}, we may compute the following variational derivatives
\begin{align*}
    \frac{\delta\ell_E}{\delta u} &= Du^\flat \otimes d^nx 
    \,,\quad
    \frac{\delta\ell_E}{\delta D} = \frac12|\bu|^2 - \pi
    \,,\quad
    \frac{\delta\ell}{\delta\pi} = (D-1)\,d^nx
    \,,\\
    \frac12\frac{\delta(\delta^2\ell_E)}{\delta(\delta u)} &= D\delta u^\flat \otimes d^nx + \delta D u^\flat \otimes d^nx
    \,,\quad
    \frac12\frac{\delta(\delta^2\ell_E)}{\delta(\delta D)} = \bu\cdot\bs{\delta u} - \delta\pi
    \,,\quad
    \frac12\frac{\delta(\delta^2\ell_E)}{\delta(\delta\pi)} = - \delta D\,d^nx  \,.
\end{align*}
We may assemble the variational derivatives of the Lagrangian into the Euler-Poincar\'e equations \eqref{eqn:EP1} and \eqref{eqn:EP2} gives
\begin{equation*}
    (\p_t + \mathcal{L}_u)(Du^\flat \otimes d^nx) = Dd\left( \frac12|\bu|^2 - \pi \right)\otimes d^nx \,,\quad\hbox{and}\quad (\p_t + \mathcal{L}_u)(D\,d^nx) = 0 \,.
\end{equation*}
The second of these equations, together with the variation in $\pi$ giving $D=1$, implies that $u$ is a divergence free vector field. The first of these equations is Euler's momentum equation in its geometric form. We may similarly assemble the variational derivatives of $\delta^2\ell / 2$ into the equations \eqref{eqn:linearised_EP_semidirect1} and \eqref{eqn:linearised_EP_semidirect2} as follows
\begin{align}
    \begin{split}
    \left( \p_t + \mathcal{L}_u \right)\left( D\delta u^\flat \otimes d^nx + \delta D u^\flat \otimes d^nx\right) &= - \mathcal{L}_{\delta u}\left( Du^\flat\otimes d^nx \right)
    \\
    &\qquad\quad + \delta Dd\left( \frac12|\bu|^2 - \pi \right)\otimes d^nx + Dd\left( \bu \cdot \bs{\delta u} - \delta\pi \right)\otimes d^nx \,,
    \end{split}
    \\
    (\p_t + \mathcal{L}_u)(\delta D\,d^nx) &= - \mathcal{L}_{\delta u}(D\,d^nx)
    \,.
\end{align}
The variations in $\pi$ and $\delta\pi$ imply that $D=1$ and $\delta D = 0 \,$. Hence $\delta u$ is a divergence free vector field and we have
\begin{equation}\label{eqn:Euler_linearised_geometric}
    (\p_t + \mathcal{L}_u)(\delta u^\flat) = - \mathcal{L}_{\delta u}u^\flat + d(\bu\cdot \bs{\delta u}- \delta \pi) \,,\quad\hbox{and}\quad (\ast d \ast)u^\flat = (\ast d \ast)\delta u^\flat = 0 \,,
\end{equation}
where $\ast$ denotes the Hodge star. In three dimensions, this can be represented in vector calculus notation as
\begin{equation}
    \p_t \bs{\delta u} - \bu\times {\rm curl}\, \bs{\delta u} = \bs{\delta u}\times {\rm curl}\,\bu - \nabla(\bu \cdot \bs{\delta u}+ \delta \pi) \,,\quad\hbox{and}\quad {\rm div}\, \bu = {\rm div}\, \bs{\delta u} = 0 \,,
\end{equation}
or, equivalently,
\begin{equation}\label{eqn:Euler_linearised_VC}
    \p_t \bs{\delta u} + \bu\cdot\nabla\bs{\delta u} + \bs{\delta u}\cdot\nabla\bu = -\nabla(\delta\pi) \,.
\end{equation}
Taking the approach described in Section \ref{subsec:variational_linearisation}, these equations, when written in terms of $\xi = (\delta g) \cdot g^{-1}$ give
\begin{equation}
    (\p_t + \ad^*_u)((\p_t\xi -\ad_u\xi)^\flat) = - \ad^*_{\p_t\xi -\ad_u\xi}u^\flat + d((\p_t\xi -\ad_u\xi)\intprod u^\flat - \delta \pi) \,.
\label{xi-eqns}
\end{equation}
This equation results from simply substituting in the relationship $\delta u = \p_t\xi - \ad_u\xi$ into the equation \eqref{eqn:Euler_linearised_geometric}. To obtain \eqref{xi-eqns} in vector calculus notation, the equation $\bs{\delta u} = \p_t\bxi + \bxi\cdot\nabla\bu - \bu\cdot\nabla\bxi$ can be considered alongside the equation \eqref{eqn:Euler_linearised_VC}.

\subsection{The stratified thermal rotating Euler equations}

We take a constant gravitational force to act in the direction of one of our coordinates, $z$, and introduce an advected parameter, $\rho$, which models thermal effects. Furthermore, we introduce a variable representing the effects of rotation, $\mathbf{R}$, which is a given function of $\mathbf{x}$ satisfying ${\rm curl}\bR(\bx) = 2\bs{\Omega}(\bx) = f(\mathbf{x})\mathbf{\hat{z}}$, where $f$ is the Coriolis parameter. The thermal rotating Euler equations then correspond to the following Lagrangian, $\ell_{tE}:\mathfrak{X}(M)\times {\rm Den}(M) \times \Lambda^0(M)\rightarrow \mathbb{R}$
\begin{align}
    \ell_{tE}(u,Dd^nx,\rho;\pi) &= \int_M D\rho \left(\frac{1}{2}|\bu|^2 + \bu\cdot\bR - gz\right) - \pi(D-1)\,d^nx \,,
    \label{eqn:thermal_Euler_Lagrangian}
    \\
    \begin{split}\label{eqn:thermal_Euler_Lagrangian_var}
    \frac12\delta^2\ell_{tE} &= \int_M \frac{D\rho}{2}|\bs{\delta u}|^2 + \left( \delta D \rho + D\delta\rho\right)\left( \bu\cdot\bs{\delta u} + \bs{\delta u}\cdot \bR \right) 
    \\
    &\qquad\quad + \delta D\delta \rho\left(\frac{|\bu|^2}{2} + \bu\cdot\bR - gz \right)- \delta\pi \delta D \,d^nx \,.
    \end{split}
\end{align}
The variational derivatives can be computed in a manner analogous to those found in Section \ref{subsec:Euler}. The Euler-Poincar\'e equations are
\begin{align}
    (\p_t+\ad^*_u)\left( D\rho(u^\flat + \bR\cdot d\bx)\otimes d^nx \right) &= Dd\left( \rho\varpi - \pi\right)\otimes d^nx - D\varpi d\rho\otimes d^nx
    \label{eqn:thermal_Euler_EP1}
    \,,\\
    (\p_t + \mathcal{L}_u)(D\,d^nx) &= 0
    \label{eqn:thermal_Euler_EP2}
    \,,\\
    (\p_t + \mathcal{L}_u)\rho &= 0
    \label{eqn:thermal_Euler_EP3}
    \,,
\end{align}
where
\begin{equation}
    \varpi := \frac{1}{2}|\bu|^2 + \bu\cdot\bR - gz \,.
\end{equation}
Making use of the advection equations and the pressure constraint $D=1$, the first of these is
\begin{equation}\label{eqn:thermal_Euler_EP4}
    \rho(\p_t + \ad^*_u)( u^\flat + \bR\cdot d\bx) = \rho d\varpi - d\pi \,. 
\end{equation}
In three dimensions, these equations can be expressed in vector calculus form as
\begin{align*}
    \p_t\bu + \bu\cdot\nabla\bu - \bu\cdot{\rm curl}\,\bR &= -g\bs{\hat z} - \frac{1}{\rho}\nabla\pi 
    \,,\\
    \p_t\rho + \bu\cdot\nabla\rho &= 0
    \,,\\
    \nabla\cdot\bu &= 0
    \,,
\end{align*}
where $\bs{\hat z}$ is the unit vector in the $z$ direction.
Similarly, the variational derivatives of the functional defined in equation \eqref{eqn:thermal_Euler_Lagrangian_var} can be substituted into the equations \eqref{eqn:linearised_EP_semidirect1} and \eqref{eqn:linearised_EP_semidirect2} to give
\begin{align}
\begin{split}
    (\p_t + \ad^*_u)&\left( D\rho\delta u^\flat\otimes d^nx + ( \delta D\rho + D\delta\rho)( u^\flat + \bR\cdot d\bx )\otimes d^nx \right) 
    \\
    &= -\mathcal{L}_{\delta u}\left(D\rho(u^\flat + \bR\cdot d\bx)\otimes d^nx \right) + \delta Dd\left( \rho\varpi - \pi \right)\otimes d^nx
    \\
    &\qquad + Dd\left( \rho(\bu\cdot\bs{\delta u} + \bs{\delta u}\cdot \bR) + \delta \rho \varpi - \delta\pi \right)\otimes d^nx
    \\
    &\qquad - D\varpi d(\delta \rho)\otimes d^nx - \left( D(\bu\cdot\bs{\delta u} + \bs{\delta u}\cdot \bR) + \delta D\varpi \right)d\rho\otimes d^nx \,,
\end{split}
    \\
    (\p_t + \mathcal{L}_u)(\delta D\,d^nx) &= -\mathcal{L}_{\delta u}(D\,d^nx)
    \,,\\
    (\p_t + \mathcal{L}_u)\delta\rho  &= -\mathcal{L}_{\delta u}\rho
    \,.
\end{align}
The first of these equations is simplified when using the constraints $D=1$ and $\delta D = 0$ and, after applying the product rule, we have
\begin{equation}
\begin{aligned}
    \rho(\p_t + \ad^*_u)\delta u^\flat + \delta \rho(\p_t + \ad^*_u)\left( u^\flat + \bR\cdot d\bx \right) &= \delta\rho \,d\varpi + \rho d\left(\bu\cdot\bs{\delta u} + \bs{\delta u}\cdot \bR \right)
    \\
    &\quad -\rho\mathcal{L}_{\delta u}\left(u^\flat + \bR\cdot d\bx \right) - d\left(\delta\pi \right) \,.
\end{aligned}
\end{equation}
Making use of the Euler-Poincar\'e equation \eqref{eqn:thermal_Euler_EP4}, this equation is
\begin{equation}\label{eqn:tEuler_linearised_geometric}
    (\p_t + \ad^*_u)\delta u^\flat = \frac{\delta\rho}{\rho^2}d\pi + d\left(\bu\cdot\bs{\delta u} + \bs{\delta u}\cdot \bR \right) - \mathcal{L}_{\delta u}\left(u^\flat + \bR\cdot d\bx \right) - \frac{1}{\rho}d\left(\delta\pi \right) \,.
\end{equation}
In vector calculus notation, the equations are given in three dimensions by
\begin{align}
    \p_t\bs{\delta u} + \bu\cdot\nabla \bs{\delta u} + \bs{\delta u}\cdot\nabla\bu - \bs{\delta u}\times {\rm curl}\,\bR &= \frac{\delta\rho}{\rho^2}\nabla\pi - \frac{1}{\rho}\nabla(\delta\pi) \label{eqn:tEuler_linearised_VC}
    \,,\\
    \p_t\delta\rho + \bu\cdot\nabla\delta\rho + \bs{\delta u}\cdot\nabla\rho &= 0
    \,,\\
    \nabla\cdot\bu = \nabla\cdot \bs{\delta u} &= 0 \,.
\end{align}
The equation for $\xi = \delta g \cdot g^{-1} = \bxi \cdot\nabla$ is given by substituting the equations
\begin{equation}\label{eqn:thermal_variations_geometric}
    \delta u = \p_t\xi -\ad_u\xi \,,\quad\hbox{and}\quad \delta\rho = - \mathcal{L}_\xi\rho \,,
\end{equation}
or
\begin{equation}\label{eqn:thermal_variations_VC}
    \bs{\delta u} = \p_t\bxi + \bxi\cdot\nabla\bu - \bu\cdot\nabla\bxi \,,\quad\hbox{and}\quad \delta\rho = -\bxi\cdot\nabla\rho \,,
\end{equation}
into equation \eqref{eqn:tEuler_linearised_geometric} or \eqref{eqn:tEuler_linearised_VC} respectively.

\subsection{The Euler-Boussinesq equations}\label{subsec:EB}
Notice that the introduction of thermal effects in the previous example results in the pressure, $\pi$, from the fluid remains in the equation governing the linear perturbation. This is not the case if one makes the Boussinesq approximation as follows. The Lagrangian is now defined to be
\begin{align}
    \ell_{EB}(u,Dd^nx,\rho;\pi) &= \int_M D \left(\frac{1}{2}|\bu|^2 + \bu\cdot\bR - g\rho z\right) - \pi(D-1)\,d^nx \,,
    \label{eqn:Euler_Boussinesq_Lagrangian}
    \\
    \frac12\delta^2\ell_{EB} &= \int_M \frac{D}{2}|\bs{\delta u}|^2 + \delta D\left( \bu\cdot\bs{\delta u} + \bs{\delta u}\cdot \bR \right) - gz \delta D\delta \rho - \delta\pi \delta D \,d^nx 
    \label{eqn:Euler_Boussinesq_Lagrangian_var}\,.
\end{align}
Proceeding as in the previous examples, the Euler-Poincar\'e equation is
\begin{equation}
    (\p_t + \ad^*_u)(u^\flat + \bR\cdot d\bx) = d\left( \frac12|\bu|^2 + \bu\cdot\bR - \pi \right) - g\rho dz \,,
\end{equation}
where the incompressibility constraint and the equation \eqref{eqn:EP1} has been rearranged after substituting in the variational derivatives of the Lagrangian, $\ell_{EB}$. In three dimensions, this equation is given by
\begin{align*}
    \p_t\bu + \bu\cdot\nabla\bu - \bu\times{\rm curl}\,\bR &= -\nabla\pi - g\rho\bs{\hat z}
    \,,\\
    \p_t\rho + \bu\cdot\nabla\rho &= 0
    \,,\\
    \nabla\cdot\bu &= 0
    \,.
\end{align*}
Computing the variational derivatives with respect to $\delta u$, $\delta D$, and $\delta\pi$ and substituting the results into equations \eqref{eqn:linearised_EP_semidirect1} and \eqref{eqn:linearised_EP_semidirect2}, we have
\begin{equation}
\begin{aligned}
    (\p_t + \ad^*_u)(D\delta u^\flat\otimes d^nx + \delta D (u^\flat+\bR\cdot d\bx)\otimes d^nx) &= -\mathcal{L}_{\delta u}(D(u^\flat+\bR\cdot d\bx)\otimes d^nx) 
    \\
    &\qquad + Dgz\,d(\delta\rho)\otimes d^nx + \delta Dgz\,d\rho\otimes d^nx
    \\
    &\qquad + \delta D d\left(\frac12|\bu|^2 + \bu\cdot\bR -g\rho z - \delta\pi \right)\otimes d^nx 
    \\
    &\qquad+ D d\left( \bu\cdot\bs{\delta u} + \bs{\delta u}\cdot\bR - gz\delta\rho - \delta\pi \right)\otimes d^nx \,.
\end{aligned}
\end{equation}
Making use of the constraints $D=1$ and $\delta D = 0$, resulting from the arbitrary variation in $\pi$ and $\delta\pi$, we have
\begin{equation}\label{eqn:EB_linearised_geometric}
    (\p_t + \ad^*_u)\delta u^\flat = -\mathcal{L}_{\delta u}(u^\flat+\bR\cdot d\bx) + gz\,d(\delta\rho) +d\left( \bu\cdot\bs{\delta u} + \bs{\delta u}\cdot\bR - gz\delta\rho - \delta\pi \right) \,.
\end{equation}
The equations, in three dimensions, are therefore
\begin{align}
    \p_t\bs{\delta u} + \bu\cdot\nabla\bs{\delta u} + \bs{\delta u}\cdot\nabla\bu - \bs{\delta u}\times {\rm curl}\,\bR &= -\nabla(\delta\pi) - g\,\delta\rho \bs{\hat z}  \label{eqn:EB_linearised_VC} 
    \,,\\
    \p_t\delta\rho + \bu\cdot\nabla\delta\rho + \bs{\delta u}\cdot\nabla\rho &= 0
    \,,\\
    \nabla\cdot\bu = \nabla\cdot \bs{\delta u} &= 0 \,.
\end{align}
As for the previous example, the equation for $\xi$ follows from substituting the equations \eqref{eqn:thermal_variations_geometric} into equation \eqref{eqn:EB_linearised_geometric} or equations \eqref{eqn:thermal_variations_VC} into equation \eqref{eqn:EB_linearised_VC}. 

\subsection{The 2D thermal rotating shallow water equations}

We here consider the two dimensional thermal rotating shallow water equations, which can be interpreted as an approximation to three dimensional models using vertical averaging. The Lagrangian is
\begin{align}
    \ell_{TRSW}(u,\eta\,d^nx,\rho) &= \int_{M\subseteq\mathbb{R}^2} \left(\frac{|\bu|^2}{2} + \bu\cdot\bR - \frac{g\rho}{2}(\eta - 2h) \right)\eta\,d^2x \,,
    \label{eqn:trsw_Lagrangian}
    \\
    \frac12\delta^2\ell_{TRSW} &= \int_{M\subseteq\mathbb{R}^2} \left( \frac{\eta}{2}|\bs{\delta u}|^2 + \delta\eta\, \bs{\delta u}\cdot(\bu + \bR) - g\,\delta\rho\,\delta\eta\,(\eta - h) - \frac{g\rho}{2}\delta\eta^2 \right)\,d^2x
    \label{eqn:trsw_Lagrangian_var}\,.
\end{align}
The variational derivatives are computed as follows
\begin{align*}
    \frac{\delta\ell}{\delta u} &= \eta (u^\flat + \bR\cdot d\bx)\otimes d^2x \,,\quad \frac{\delta\ell}{\delta \rho} = - \frac{g\eta}{2}(\eta - 2h)\,d^2x \,,\quad  \frac{\delta\ell}{\delta \eta} = \frac{|\bu|^2}{2} + \bu\cdot\bR - g\rho(\eta - h) \,,
    \\
    \frac12\frac{\delta(\delta^2\ell)}{\delta(\delta u)} &= (\eta \,\delta u^\flat + \delta\eta (u^\flat + \bR\cdot d\bx))\otimes d^2x \,,\quad      \frac12\frac{\delta(\delta^2\ell)}{\delta(\delta \rho)} = -g\,\delta\eta(\eta - h)\,d^2x \,,
    \\
    \frac12\frac{\delta(\delta^2\ell)}{\delta(\delta \eta)} &= \bs{\delta u}\cdot(\bu + \bR) -g\,\delta\rho(\eta - h) - g\rho\,\delta\eta \,.
\end{align*}
Assembling these into the Euler-Poincar\'e equation, making use of the fact that $\eta$ and $\rho$ are advected quantities, we have
\begin{equation}\label{eqn:TRSW_EP}
    (\p_t + \ad^*_u)(u^\flat + \bR\cdot d\bx) = d\left(\frac{|\bu|^2}{2} + \bu\cdot\bR \right) -g\rho\,d(\eta-h) - \frac{g\eta}{2}\,d\rho \,,
\end{equation}
which, in vector calculus notation, is
\begin{align}
    \p_t\bu + \bu\cdot\nabla\bu - (\nabla^\perp \cdot\bR)\bu^\perp = -g\rho\nabla(\eta-h) - \frac12 g\eta\nabla\rho 
    \,,\\
    \p_t\rho + \bu\cdot\nabla\rho &= 0
    \,,\\
    \p_t\eta + \nabla\cdot (\eta\bu) &= 0
    \,.
\end{align}
Assembling the variational derivatives into the equations \eqref{eqn:linearised_EP_semidirect1} and \eqref{eqn:linearised_EP_semidirect2}, we have
\begin{equation}
\begin{aligned}
    (\p_t + \ad^*_u)\left( (\eta \,\delta u^\flat + \delta\eta (u^\flat + \bR\cdot d\bx))\otimes d^2x \right) &= -\mathcal{L}_{\delta u}(\eta (u^\flat + \bR\cdot d\bx)\otimes d^2x)
    \\
    &\quad + \delta\eta\,d\left( \frac{|\bu|^2}{2} + \bu\cdot\bR - g\rho(\eta - h) \right)\otimes d^2x
    \\
    &\quad + \eta\,d\left( \bs{\delta u}\cdot(\bu + \bR) -g\,\delta\rho(\eta - h) - g\rho\,\delta\eta \right) \otimes d^2x
    \\
    &\quad + \frac{g\eta}{2}(\eta - 2h)\,d(\delta\rho)\otimes d^2x
    \\
    &\quad + g\,\delta\eta(\eta - h)\,d\rho\otimes d^2x \,.
\end{aligned}
\end{equation}
These equations can be simplified by applying the product rule and making use of the Euler-Poincar\'e equation \eqref{eqn:TRSW_EP} to give
\begin{equation}\label{eqn:TRSW_linearised_geometric}
\begin{aligned}
    (\p_t + \ad^*_u)\delta u^\flat &= -\mathcal{L}_{\delta u}(u^\flat + \bR\cdot d\bx) + d\left( \frac{|\bu|^2}{2} + \bu\cdot\bR \right)
    \\
    &\quad -g\,\delta\rho\,d(\eta - h) - g\rho\,d(\delta\eta)
-\frac{g\eta}{2}\,d(\delta\rho) - \frac{g}{2}\delta\eta\, d\rho \,. 
\end{aligned}
\end{equation}
The equations, in vector calculus form, are
\begin{align}
    \p_t\bs{\delta u} + \bu\cdot\nabla\bs{\delta u} + \bs{\delta u}\cdot\nabla\bu - (\nabla^\perp \cdot\bR)\bs{\delta u}^\perp &= \frac12g\,\delta\eta\nabla\rho - \frac12g\eta\nabla(\delta\rho) - g\,\delta\rho\nabla(\eta-h) - g\rho\nabla(\delta\eta)
    \label{eqn:TRSW_linearised_VC}
    \,,\\
    \p_t\delta\rho + \bu\cdot\nabla\delta\rho + \bs{\delta u}\cdot\nabla\rho &= 0
    \,,\\
    \p_t\delta\eta + \nabla\cdot(\delta\rho\bu) + \nabla\cdot(\rho\bs{\delta u}) &= 0
    \,.
\end{align}
The equation for $\xi = \delta g\cdot g^{-1}$ is given, in its geometric form, by substituting the equations
\begin{equation}
    \delta u = \p_t\xi -\ad_u\xi \,,\quad \delta\rho = - \mathcal{L}_\xi\rho \,,\quad\hbox{and}\quad \delta\eta = -\mathcal{L}_\xi\eta \,,
\end{equation}
into equation \eqref{eqn:TRSW_linearised_geometric}. Alternatively, in vector calculus notation, this equation corresponds to substituting
\begin{equation}
    \bs{\delta u} = \p_t\bxi + \bxi\cdot\nabla\bu - \bu\cdot\nabla\bxi \,,\quad \delta\rho = -\bxi\cdot\nabla\rho \,,\quad\hbox{and}\quad \delta \eta = - \nabla\cdot(\eta\bxi) \,,
\end{equation}
into equation \eqref{eqn:TRSW_linearised_VC}.

\section{Numerical simulations}\label{sec:numerics}
In this section, we consider numerical simulations of the linearised Euler-Poincar\'e equations and the associated dynamics of the perturbation vector field for the example of the Euler-Boussinesq (EB) equations given in Section \ref{subsec:EB}. Simplifying to a 2D vertical domain, the incompresibility conditions of both $\bu$ and $ \bs{\delta u} $ allow us to express the EB equations and their linearised equations in streamfunction and vorticty formulation. Using the Jacobian operator $\mcal{J}: C^\infty(M)\times C^\infty(M) \rightarrow \mathbb{R}$ in the vertical $(x, z)$-plane defined by
\begin{align}
    \mcal{J}(f, g) = - \hat{\bs{y}}\cdot \nabla f \times \nabla g = \p_x f \p_z g - \p_z f \p_x g \,,
\end{align}
we have the following equivalent formulation of equation \eqref{eqn:EB_linearised_VC}
\begin{align}
    \begin{split}
        &\p_t\omega + \mcal{J}(\psi, \omega) + g\mcal{J}(z, b) = 0\,,\\
        &\p_t b + \mcal{J}(\psi, b) = 0\,,\quad \text{where}\quad \omega := \triangle \psi \quad \text{and} \quad \bu = \nabla^\perp \psi\,,
    \end{split}\\
    \begin{split}
        &\p_t\delta \omega + \mcal{J}(\psi, \delta \omega) + \mcal{J}(\delta \psi, \omega) + g\mcal{J}(z, \delta b) = 0\,,\\
        &\p_t \delta b + \mcal{J}(\psi, \delta b) + \mcal{J}(\delta \psi, b) = 0\,,\quad \text{where}\quad \delta \omega := \triangle \delta \psi \quad \text{and} \quad \bs{\delta u} = \nabla^\perp \delta \psi\,.
    \end{split}
\end{align}
Substituting the Ansatz $\delta b = -\mathcal{L}_\xi b = -\bs{\xi}\cdot \nabla b$ where $\bs{\xi} \in\ \mathfrak{X}(M)$ is the perturbation vector field, we have the equivalent form of the linearised dynamics as
\begin{align}
    \begin{split}
        &\p_t\delta \omega + \mcal{J}(\psi, \delta \omega) + \mcal{J}(\delta \psi, \omega) - g\mcal{J}(z, \bs{\xi}\cdot\nabla b) = 0\,,\\
        &\p_t \bs{\xi} + (\bs{\xi}\cdot \nabla)\bu -\bu\cdot \nabla \bs{\xi} = \bs{\delta u} \,,\quad \text{where}\quad \delta \omega := \triangle \delta \psi \quad \text{and}\quad \bs{\delta u} = \nabla^\perp \delta \psi \,.
    \end{split}
\end{align}
The configurations of the numerical simulations are as follows. The computational domain is $\Omega = [0,1]\times [0,1]$ which is discretized using $256\times 256$ finite element cells. The boundary conditions for $\bu$ and $\bs{\delta u}$ are periodic in the $x$ direction and free slip in $z$. These boundary conditions can be enforced through their definition from the stream functions $\psi$ and $\delta \psi$ respectively, both of which have homogeneous Dirichlet boundary conditions in $y = 0, 1$ and are periodic in $x$. In the absence of viscosity, no other boundary conditions are required. The fluid vorticity $\omega$, perturbed vorticity $\delta \omega$ and buoyancy $b$ are approximated with the $2^{nd}$ order discontinuous Galerkin finite element space ($DG1$); the stream function $\psi$ and the perturbation stream function $\delta \psi$ are approximated with the $2^{nd}$ order continuous Galerkin finite element space ($CG1$); Lastly, the perturbation vector field $\bs{\xi}$, fluid velocity $\bu$, and perturbed fluid velocity $\bs{\delta u}$ are approximated with the vectorised $2^{nd}$ order continuous Galerkin finite element space. The numerical method is implemented using the firedrake software \cite{firedrake-manual} and we ran the simulation for a total of $16$ time units. The snapshots of interest are presented in Figures \ref{fig:pv_5} - \ref{fig:xi_8}.
\begin{figure}
    \centering
    \includegraphics[width=0.32\textwidth]{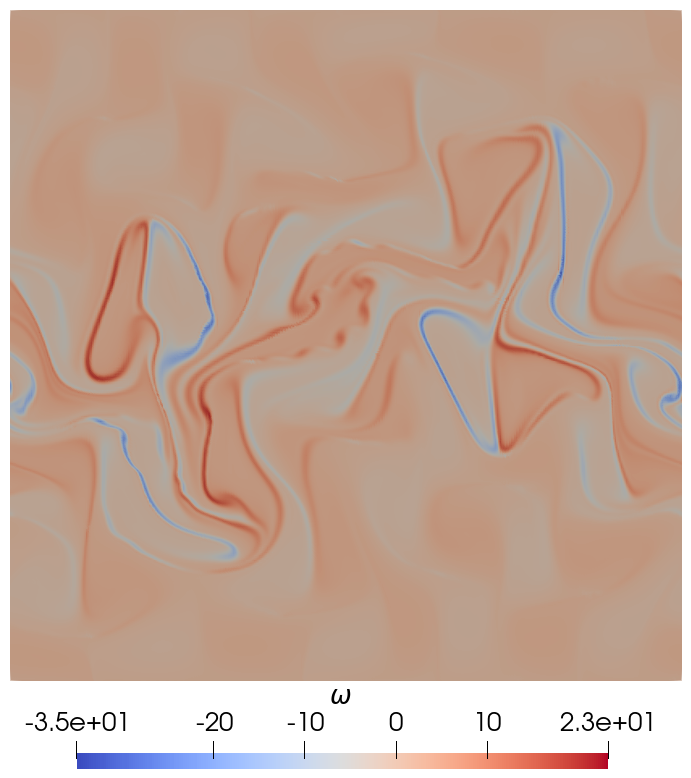}
    \includegraphics[width=0.32\textwidth]{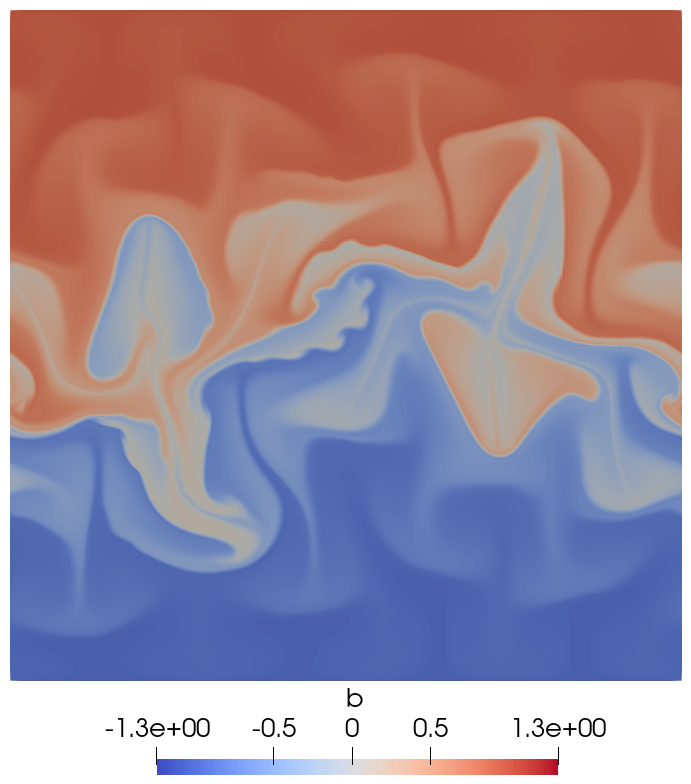}
    \includegraphics[width=0.32\textwidth]{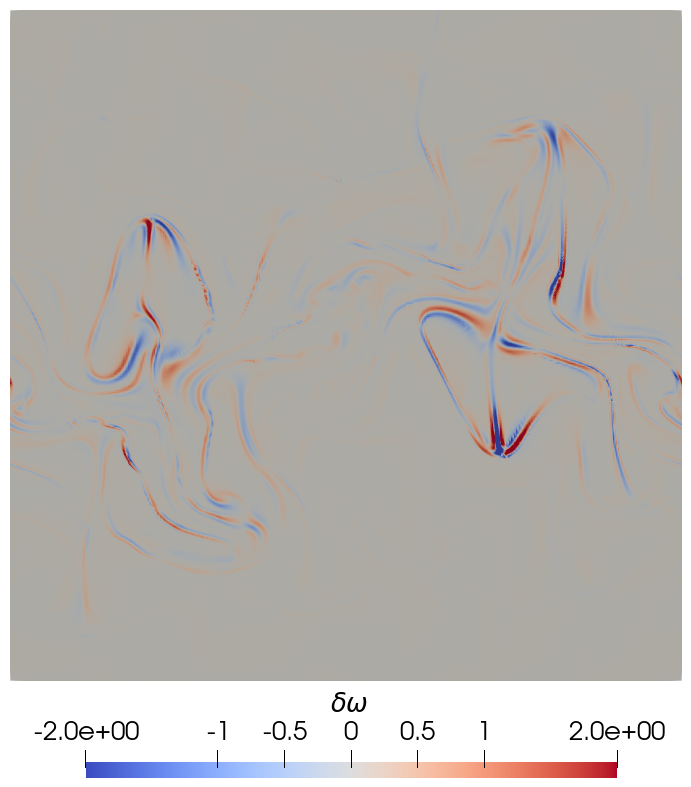}
    \caption{At $t=5$, one observes the initial phases of Kelvin Helmholtz instabilities generated by the initial conditions in fluid vorticity and buoyancy gradients in the snapshots of $\omega$ (left) and $b$ (middle). The inhomogeneous regions in the $\delta w$ snapshot (right) closely track the fronts of these instabilities, notably the downwards plume and upwards plume located at the left and right side of domain centre. This feature allows the perturbation vorticity $\delta \omega$ to be used a diagnostic for instabilities.}
    \label{fig:pv_5}
\end{figure}

\begin{figure}
    \centering
    \includegraphics[width=0.32\textwidth]{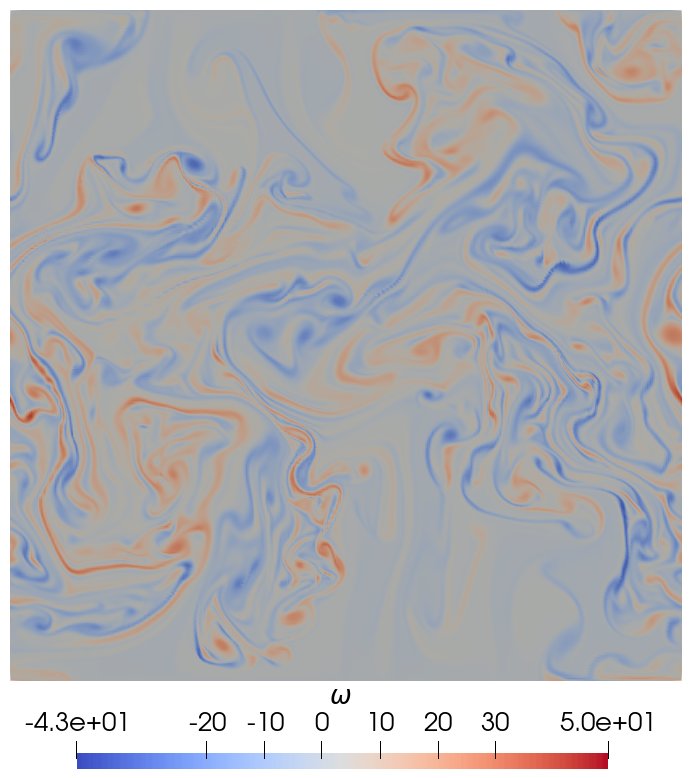}
    \includegraphics[width=0.32\textwidth]{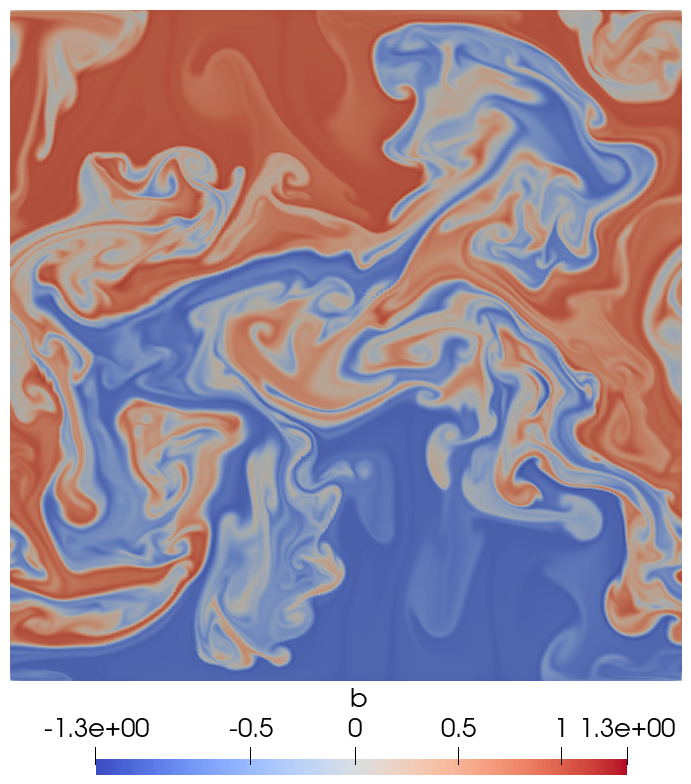}
    \includegraphics[width=0.32\textwidth]{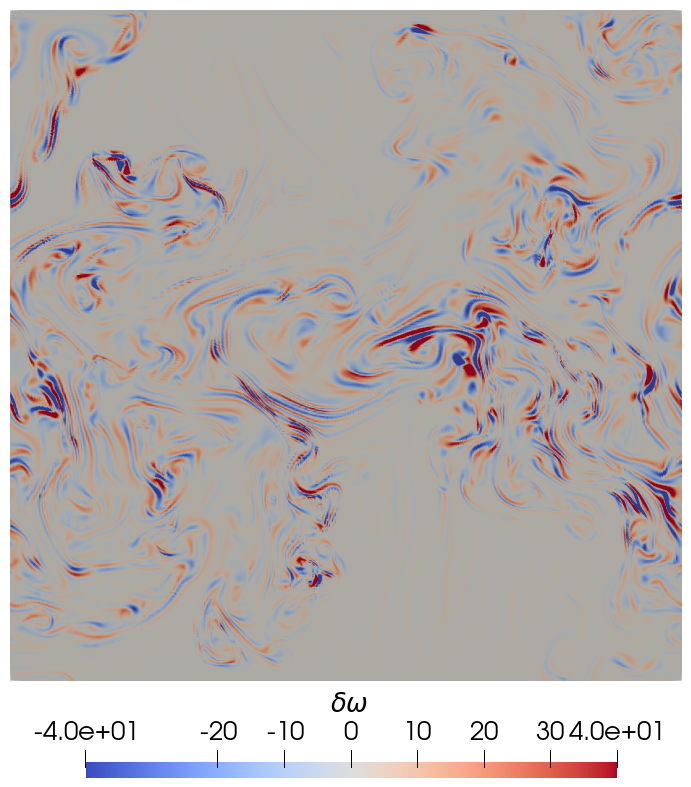}
    \caption{At $t=8$, the Kelvin Helmholtz instabilities are fully developed as they are visible in the $\omega$ (left) and $b$ (middle) snapshots. The perturbation vorticity $\delta \omega$ (right) shows strong correlations with these instabilities whilst the magnitude are $\approx 10$ times larger than the $\delta \omega$ snapshot at $t = 5$. We note that the region of largest $\delta \omega$ originated from the initial Kelvin Helmholtz instability shown in Figure \ref{fig:pv_5}. This is indeed expected as $\delta \omega$ is advected by the fluid velocity, $\bu$, with an additional forcing term, $\bs{\delta u}\cdot \nabla \omega$, that generates circulation of $\bs{\delta \bu}$.}
    \label{fig:pv_8}
\end{figure}

\begin{figure}
    \centering
    \includegraphics[width=0.30\textwidth]{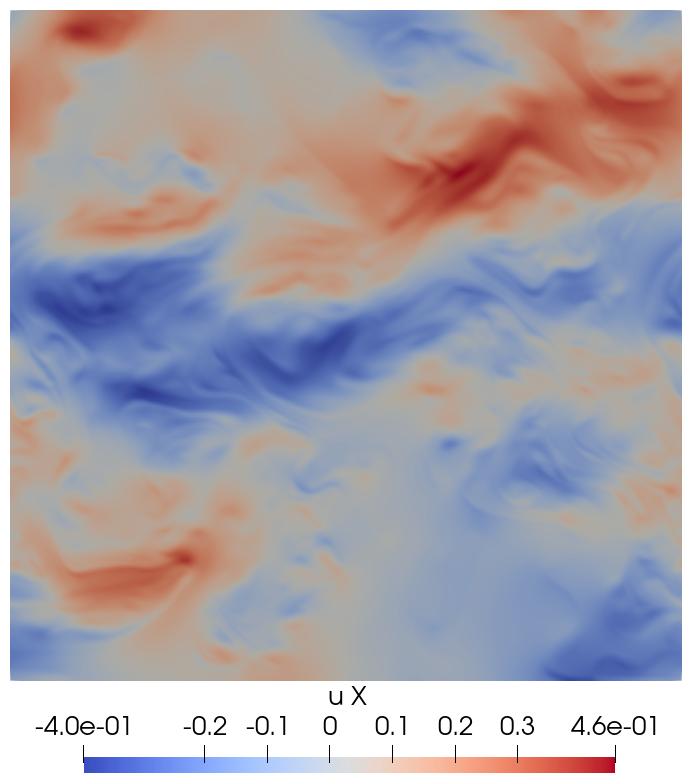}
    \includegraphics[width=0.30\textwidth]{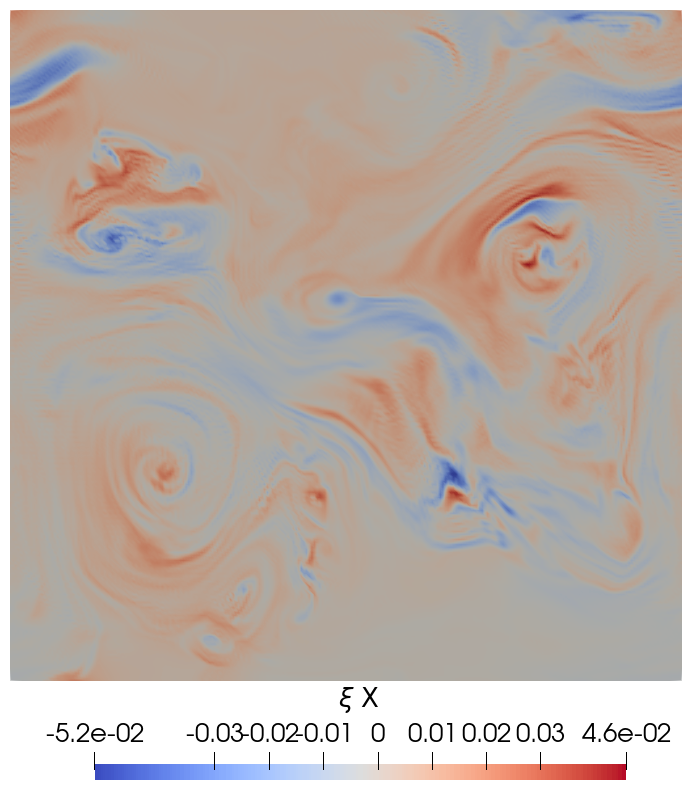}
    \includegraphics[width=0.30\textwidth]{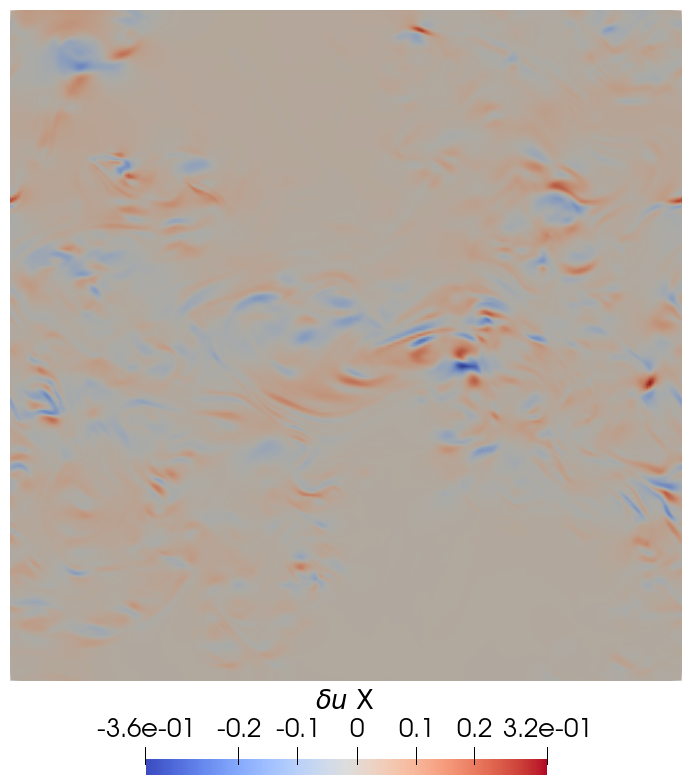}
    
    \includegraphics[width=0.30\textwidth]{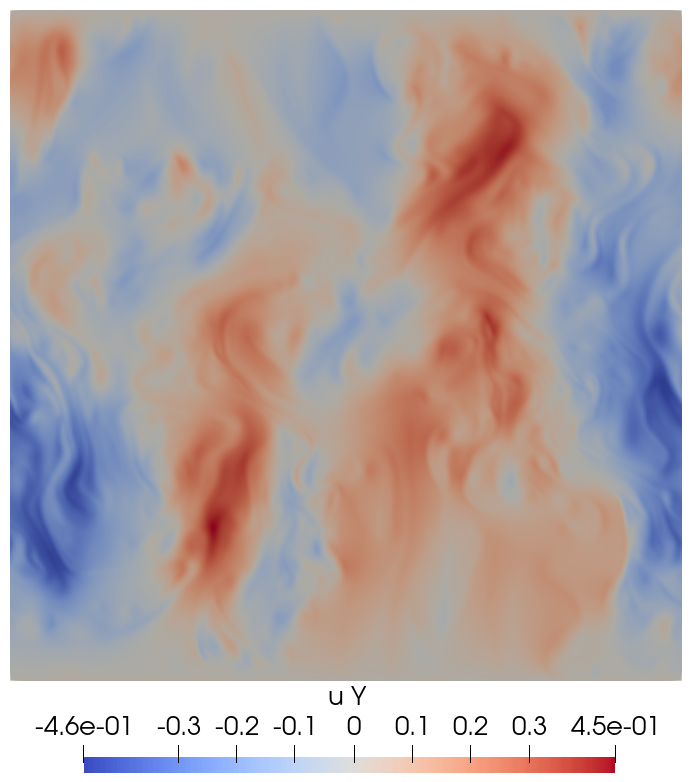}
    \includegraphics[width=0.30\textwidth]{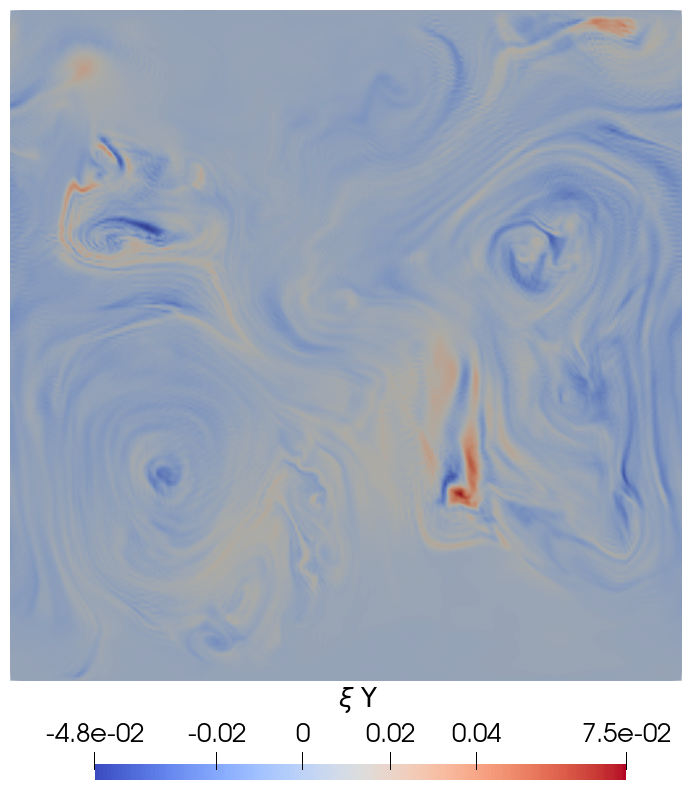}
    \includegraphics[width=0.30\textwidth]{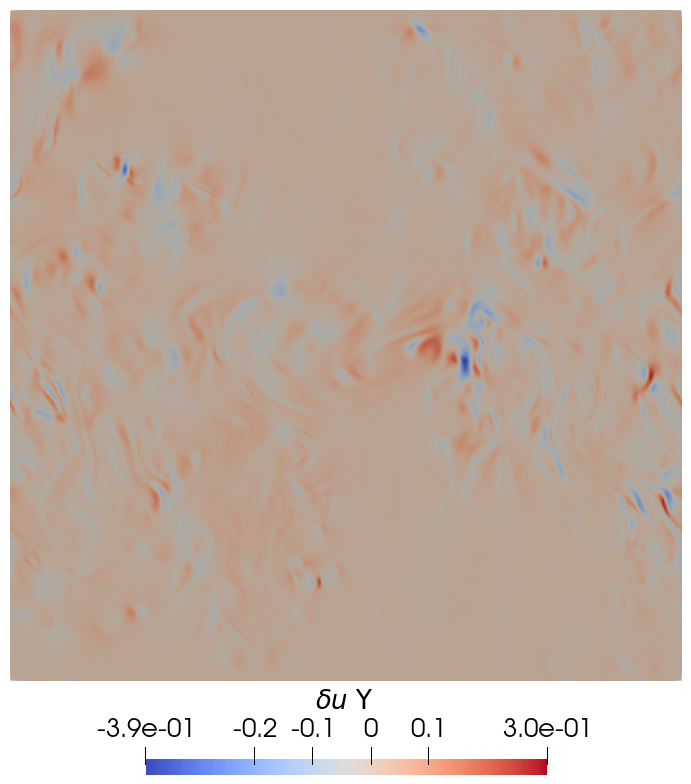}
    
    \caption{Snapshots of the fluid velocity field $\bu$ (left), perturbation vector field $\xi$ (middle), and the perturbed velocity field $\bs{\delta u}$ (right) at time $t= 8$.}
    \label{fig:xi_8}
\end{figure}

\section{Summary, open problems and outlook}\label{sec4-sumout}
The present paper has treated higher-order variations in the Euler-Poincar\'e variational principles with applications to ideal fluid dynamics. The 1$^{st}$ variations yield the well known models of ideal fluid dynamics via the Euler-Poincar\'e approach \cite{HMR1998}. The 2$^{nd}$ variations yield equations for linear perturbations propagating in the frame of the fluid motion arising from the 1$^{st}$ variation. The advantage of this approach is that the stability equations arising from the 2$^{nd}$ variation are perturbation equations for time-dependent flows, not only for equilibrium time-independent flows, although fluid equilibrium solutions are permitted. For the EPDiff equation for geodesics on a Lie group, the linearied equation derived using the methods discussed in this paper is shown to be equivalent to Jacobi's equation and thus the arbitrary vector field employed in the 1$^{st}$ variation is the Jacobi field. The physical examples of fluid models given in this paper are expressed naturally on semidirect product spaces and are thus not geodesic equations. It remains to understand how the broken symmetry of ideal fluid dynamics with advected quantities influences the Jacobi equation for the vector field $\xi = \delta gg^{-1}$, when expressed in terms of covariant derivatives, and hence the behaviour of nearby trajectories in the Lie group.



\subsection*{Acknowledgements}
We are grateful to C. Cotter, D. Crisan, J.-M. Leahy, A. Lobbe, J. Woodfield, as well as H. Dumpty for several thoughtful suggestions during the course of this work which have improved or clarified the interpretation of its results.
DH and RH were partially supported during the present work by Office of Naval Research (ONR) grant award  N00014-22-1-2082, ``Stochastic Parameterization of Ocean Turbulence for Observational Networks''. DH and OS were partially supported during the present work by European Research Council (ERC) Synergy grant ``Stochastic Transport in Upper Ocean Dynamics" (STUOD) -- DLV-856408.


\begin{thebibliography}{}

\bibitem{Arnold1965}
Arnold, V.I., 1965. Conditions for nonlinear stability of the stationary plane curvilinear flows of an ideal fluid. Doklady Mat. Nauk., 162 (5), pp. 773-777

\bibitem{Arnold1966}
Arnold, V.I., 1966. Sur la g\'eométrie diff\'erentielle des groupes de Lie de dimension infinie et ses applications \`a l'hydrodynamique des fluides parfaits. In Annales de l'institut Fourier (Vol. 16, No. 1, pp. 319-361).

\bibitem{AK1998}
Arnold, V.I. and Khesin, B., 1998. Topological Methods in Hydrodynamics, Springer, New York.

\bibitem{BKMR1996}
Bloch, A., Krishnaprasad, P.S., Marsden, J.E., Ratiu T.S., 1996. The Euler-Poincaré equations and double bracket dissipation. Commun.Math. Phys. 175, 1–42. \url{https://doi.org/10.1007/BF02101622}

\bibitem{CF1996}
Casciaro, B. and Francaviglia, M., On the second variation for first order calculus of variations on fibered manifolds. I: Generalized Jacobi equations, Rend. Mat., Serie VII,16 (1996), pp. 233–264.

\bibitem{Chandra1969}
Chandrasekhar, S., 1969. Ellipsoidal Figures of Equilibrium (New Haven. Conn.: Yale University).

\bibitem{Chandra1970}
Chandrasekhar, S., 1970. Solutions of two problems in the theory of gravitational radiation. Phys. Rev. Lett., 24(11), p.611.

\bibitem{Chandra1983}
Chandrasekhar, S., 1983. The Mathematical Theory of Black Holes. Oxford University Press, Oxford, 1983, xxii / 646 pp.

\bibitem{Chen-etal1998}
Chen, S. Foias, C. Holm, D.D. Olson, E.J., Titi, E.S.  and Wynne, S., 1998.
The Camassa-Holm equations as a closure model for
turbulent channel and pipe flows,
\textit{ Phys. Rev. Lett.}, \textbf{81} (1998) 5338-5341, 
\url{https://doi.org/10.1103/PhysRevLett.81.5338}

\bibitem{Chiaffredo-etal2023}
Chiaffredo, F., Fatibene, L., Ferraris, M., Ricossa, E. and Usseglio, D., 2023. A variational framework for higher order perturbations. arXiv preprint arXiv:2310.12907.


\bibitem{FHT2001} 
Foias, C., Holm, D.D. and Titi, E.S.,  2001.
The Navier-Stokes-alpha model of fluid turbulence. 
Physica D (152) 505-519.
\url{https://doi.org/10.1016/S0167-2789(01)00191-9}

\bibitem{FHT2002} 
Foias, C., Holm, D.D. and Titi, E.S.,  2002.
The three dimensional viscous Camassa-Holm equations, and
their relation to the Navier-Stokes equations and turbulence theory.
J. Dyn. and Diff. Eqns. (14) 1-35.
\url{https://doi.org/10.1023/A:1012984210582}

\bibitem{firedrake-manual}
Ham D. A., Kelly P.H.J., Mitchell L., Cotter C.J., et al., Firedrake User Manual. Imperial College London and University of Oxford and Baylor University and University of Washington, 2023.
\url{10.25561/104839}

\bibitem{Hecht-etal2008a}
Hecht, M.W., Holm, D.D. Petersen, M.R. and  Wingate, B.A.,
The LANS-alpha and Leray turbulence parameterizations in primitive equation ocean modeling.
J. Phys. A: Math. Theor. \textbf{41} (2008) 344009 
\url{https://doi.org/10.1088/1751-8113/41/34/344009}

\bibitem{Hecht-etal2008b}
Implementation of the LANS-alpha turbulence model in a primitive equation ocean model.
MW Hecht, DD Holm, MR Petersen, BA Wingate,
J. Comp. Physics (227) (2008) 5691.\\
\url{https://doi.org/10.1016/j.jcp.2008.02.018}

\bibitem{Hecht-etal2008c}
Efficient form of the LANS-alpha turbulence model in a primitive-equation ocean model.
\\ MW Hecht, DD Holm, MR Petersen, BA Wingate,
J. Comp. Physics (227) (2008) 5717.
\url{https://doi.org/10.1016/j.jcp.2008.02.017}

\bibitem{Holm1999PhysD}
Holm, D.D., 1999. Fluctuation effects on 3D Lagrangian mean and Eulerian mean fluid motion. 
Physica D: Nonlinear Phenomena, 133(1-4), pp.215-269.\\
\url{https://doi.org/10.1016/S0167-2789(02)00552-3}

\bibitem{Holm2002Chaos}
Holm, D.D., 2002. Lagrangian averages, averaged Lagrangians, and the mean effects of fluctuations in fluid dynamics. Chaos, 12(2), pp.518-530. \url{https://doi.org/10.1063/1.1460941}

\bibitem{Holm2002PhysD}
Holm, D.D., 2002. Averaged Lagrangians and the mean effects of fluctuations in ideal fluid dynamics. Physica D: Nonlinear Phenomena, 170(3-4), pp.253-286.\\
\url{https://doi.org/10.1016/S0167-2789(02)00552-3}





\bibitem{HMR1998}Holm, D.D., Marsden, J.E. and Ratiu, T.S., 1998. The Euler–Poincaré equations and semidirect products with applications to continuum theories. Advances in Mathematics, 137(1), pp.1-81.\\
\url{https://doi.org/10.1006/aima.1998.1721}

\bibitem{HMR1998b}
Holm, D.D., Marsden, J.E., Ratiu, T.S., 1998.
Euler--Poincar\'e models of ideal fluids with nonlinear dispersion.
Phys. Rev. Lett., (80) 4173-4177.\\
\url{https://doi.org/10.1103/PhysRevLett.80.4173}

\bibitem{HMRW1985}
Holm, D.D., Marsden, J.E., Ratiu, T.S. and Weinstein, A., 1985. Nonlinear stability of fluid and plasma equilibria. Physics reports, 123(1-2), pp.1-116. \\
\url{https://doi.org/10.1016/0370-1573(85)90028-6}



\bibitem{JJost2023}
Joharinad, P. and Jost, J., 2023. Metric Spaces and Manifolds. In Mathematical Principles of Topological and Geometric Data Analysis (pp. 115-164). Cham: Springer International Publishing.

\bibitem{Jost2008}
Jost, J., 2008. Geodesics and Jacobi Fields. Riemannian Geometry and Geometric Analysis, pp.179-241.


\bibitem{Michor2006}
Michor, P.W. (2006). Some geometric evolution equations arising as geodesic equations on groups of diffeomorphisms including the Hamiltonian approach. In: Bove, A., Colombini, F., Del Santo, D. (eds) Phase Space Analysis of Partial Differential Equations. Progress in Nonlinear Differential Equations and Their Applications, vol 69. Birkh\"auser Boston. \url{https://doi.org/10.1007/978-0-8176-4521-2_11}.

\bibitem{M2015}
Michor, P.W., 2015. Manifolds of mappings and shapes. Arxiv. \url{https://arxiv.org/abs/1505.02359}

\bibitem{ModinPerrot2023}
Modin, K. and Perrot, M., 2023. Eulerian and Lagrangian stability in Zeitlin's model of hydrodynamics. arXiv preprint arXiv:2305.08479.

\bibitem{Preston2004}
Preston, S.C., 2004. For ideal fluids, Eulerian and Lagrangian instabilities are equivalent. Geometric and Functional Analysis, 14(5), pp.1044-1062.


\bibitem{Sreenivasan2019}
Sreenivasan, K. R., 2019. Chandrasekhar's Fluid Dynamics. Annual Review of Annual Review of Fluid Mechanics, 51:1-24. \url{https://doi.org/10.1146/annurev-fluid-010518-040537}

\bibitem{WashabaughPreston2017}
Washabaugh, P. and Preston, S.C.,
The geometry of axisymmetric ideal fluid flows with swirl
 - Arnold Mathematical Journal, 2017 - Springer

\bibitem{Younes2007}
Younes, L., 2007. Jacobi fields in groups of diffeomorphisms and applications. Quarterly of applied mathematics, pp.113-134.

\end{thebibliography}
\end{document}